\documentclass{llncs}
\usepackage{graphicx}
\usepackage{xspace}
\usepackage{thmtools} 
\usepackage{thm-restate}
\usepackage{todonotes}
\usepackage{enumerate}
\usepackage{soul}
\usepackage{hyperref}
\usepackage{xcolor}

\usepackage{tabularx}
\usepackage[figure,table]{hypcap}
\usepackage{multirow}
\usepackage{multicol}

\usepackage{complexity,mathtools}

\renewcommand{\orcidID}[1]{\href{https://orcid.org/#1}{\includegraphics[scale=.03]{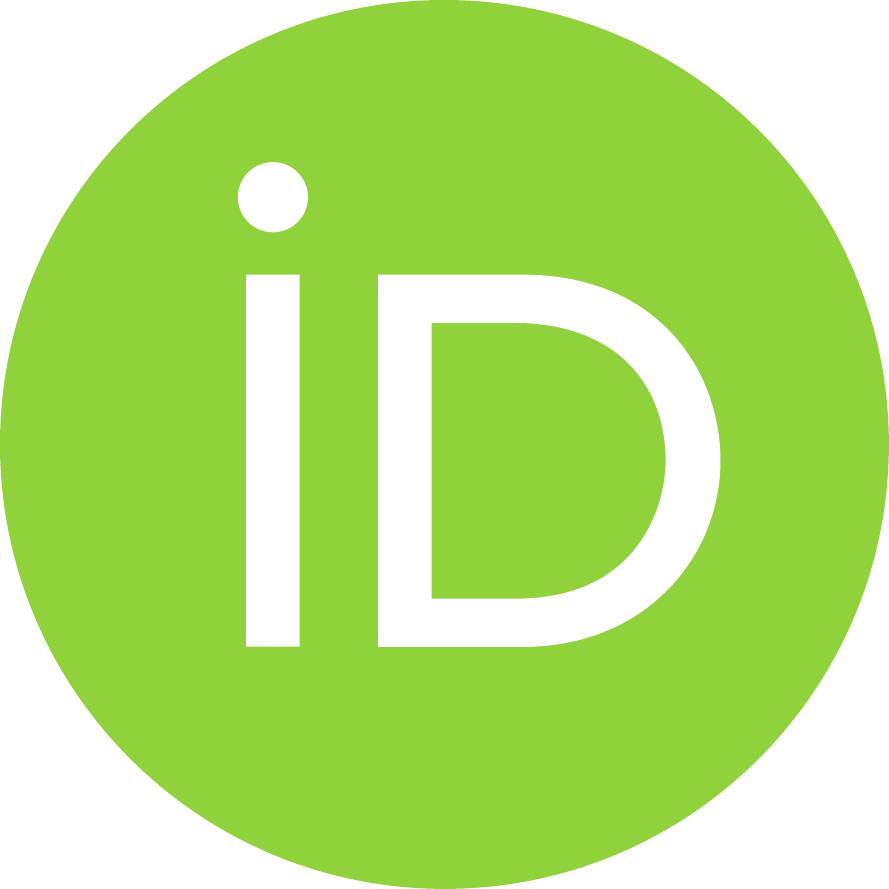}}} 
\usepackage[capitalise,nameinlink]{cleveref}
\Crefname{observation}{Observation}{Observations}
\Crefname{proposition}{Proposition}{Propositions}
\Crefname{claim}{Claim}{Claims}
\Crefname{property}{Property}{Properties}

\Crefname{enumi}{Property}{Properties}
\usepackage{paralist}
\usepackage{subfigure}
\usepackage{amssymb}
\usepackage{cite} 

\newenvironment{sketch}{\noindent\textit{Sketch of proof:}}{\hfill {\footnotesize $\square$}\bigskip}

\hypersetup{
    colorlinks=true,
    linkcolor=blue,
    filecolor=violet,  
    citecolor=violet,
    urlcolor=cyan,
    pdftitle={st-orientations-with-few-transitive-edges},
}
\graphicspath{{figures/}}

\definecolor{MyBlue}{rgb}{0.022,0.263,0.394}
\renewcommand{\emph}[1]{{\color{MyBlue}{\em #1}}\xspace}

\title{$st$-Orientations with Few Transitive Edges\thanks{Work partially supported by: (i) MIUR, grant 20174LF3T8 AHeAD: efficient Algorithms for HArnessing networked Data", (ii) Dipartimento di Ingegneria, Universita degli Studi di Perugia, grant RICBA21LG: Algoritmi, modelli e sistemi per la rappresentazione visuale di reti.}}

\author{Carla Binucci\inst{1}\orcidID{0000-0002-5320-9110}$^c$ \and
Walter Didimo\inst{1}\orcidID{0000-0002-4379-6059} \and Maurizio Patrignani\inst{2}\orcidID{0000-0001-9806-7411}}

\institute{}

\institute{Universit\`a degli Studi di Perugia, Italy\\
\email{\{carla.binucci,walter.didimo\}@unipg.it}
\and  Roma Tre University, Rome, Italy\\
\email{maurizio.patrignani@uniroma3.it}
}

\begin{document}

\maketitle

\begin{abstract}
The problem of orienting the edges of an undirected graph such that the resulting digraph is acyclic and has a single source $s$ and a single sink $t$ has a long tradition in graph theory and is central to many graph drawing algorithms. Such an orientation is called an $st$-orientation. 
We address the problem of computing $st$-orientations of undirected graphs with the minimum number of transitive edges. We prove that the problem is NP-hard in the general case. For planar graphs we describe an ILP model that is fast in practice. We experimentally show that optimum solutions dramatically reduce the number of transitive edges with respect to unconstrained $st$-orientations computed via classical $st$-numbering algorithms. Moreover, focusing on popular graph drawing algorithms that apply an $st$-orientation as a preliminary step, we show that reducing the number of transitive edges leads to drawings that are much more compact.  

\end{abstract}

%\keywords{ xx \and xx \and xx }

%\clearpage
%\setcounter{tocdepth}{3}
%\tableofcontents 

\section{Introduction}\label{se:introduction}
The problem of orienting the edges of an undirected graph in such a way that the resulting digraph satisfies specific properties has a long tradition in graph theory and represents a preliminary step of several graph drawing algorithms.
For example, Eulerian orientations require that each vertex gets equal in-degree and out-degree; they are used to compute 3D orthogonal graph drawings~\cite{DBLP:journals/dam/EadesSW00} and right-angle-crossing drawings~\cite{DBLP:journals/jgaa/AngeliniCDFBKS11}. Acyclic orientations require that the resulting digraph does not contain directed cycles (i.e., it is a DAG); they can be used as a preliminary step to compute hierarchical and upward drawings that nicely represent an undirected graph, or a partially directed graph, so that all its edges monotonically flow in the same direction~\cite{DBLP:journals/cj/BinucciD16,DBLP:journals/tcs/BinucciDP14,DBLP:reference/algo/Didimo16,DBLP:journals/jgaa/EiglspergerKE03,DBLP:journals/jgaa/FratiKPTW14,DBLP:reference/crc/HealyN13}. 

Specific types of acyclic orientations that are central to many graph algorithms and applications are the so called \emph{$st$-orientations}, also known as \emph{bipolar orientations}~\cite{DBLP:journals/dcg/RosenstiehlT86}, whose resulting digraphs have a single source~$s$ and a single sink~$t$. It is well known that an undirected graph~$G$ with prescribed vertices~$s$ and~$t$ admits an $st$-orientation if and only if $G$ with the addition of the edge $(s,t)$ (if not already present) is biconnected. The digraph resulting from an $st$-orientation is also called an \emph{$st$-graph}.
An $st$-orientation can be computed in linear time via an $st$-numbering (or $st$-ordering) of the vertices of $G$~\cite{DBLP:journals/tcs/EvenT76,DBLP:conf/esa/Brandes02}, by orienting each edge from the end-vertex with smaller number to the end-vertex with larger number~\cite{DBLP:conf/esa/Brandes02}. In particular, if $G$ is planar, a \emph{planar $st$-orientation} of $G$ additionally requires that $s$ and $t$ belong to the external face in some planar embedding of the graph. Planar $st$-orientations were originally introduced in the context of an early planarity testing algorithm~\cite{lec-67}, and are largely used in graph drawing to compute different types of layouts, including visibility representations, polyline drawings, dominance drawings, and orthogonal drawings (refer to~\cite{DBLP:books/ph/BattistaETT99,DBLP:conf/dagstuhl/1999dg}). Planar $st$-orientations and related graph layout algorithms are at the heart of several graph drawing libraries and software (see, e.g.,~\cite{DBLP:reference/crc/ChimaniGJKKM13,DBLP:reference/crc/BattistaD13,DBLP:books/sp/04/WieseE004,DBLP:books/sp/Juenger04}). Algorithms that compute $st$-orientations with specific characteristics (such as bounds on the length of the longest path) are also proposed and experimented in the context of visibility and orthogonal drawings~\cite{DBLP:journals/tcs/PapamanthouT08,DBLP:journals/jgaa/PapamanthouT10}.

\begin{figure}[tb]
	\centering
	\subfigure[8 transitive edges]{
		\includegraphics[width=0.4\textwidth]{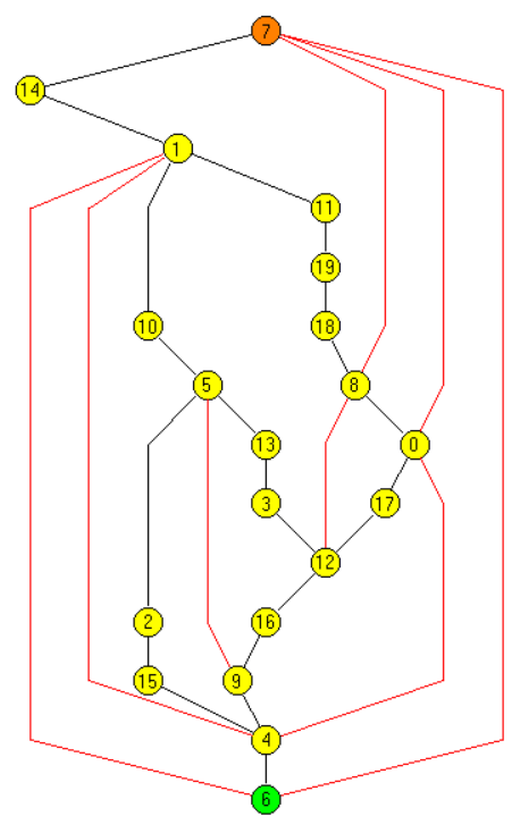}
		\label{fi:intro-a}
	}
	\hfil
	\subfigure[4 transitive edges]{
		\includegraphics[width=0.41\textwidth]{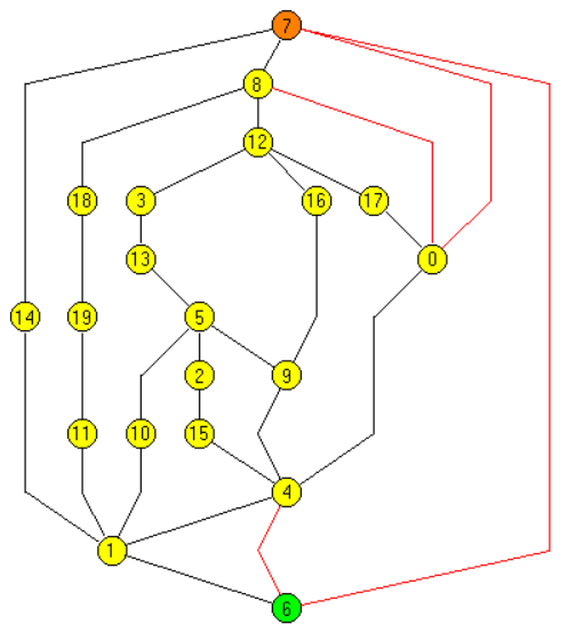}
		\label{fi:intro-b}
	}
	\caption{Two polyline drawings of the same plane graph, computed using two different $st$-orientations, with $s=6$ and $t=7$; transitive edges are in red. (a) An unconstrained $st$-orientation with $8$ transitive edges, computed through an $st$-numbering; (b) An $st$-orientation with the minimum number (four) of transitive edges; the resulting drawing is more compact and has shorter edges. \label{fi:intro}}
\end{figure}

\medskip
Our paper focuses on the computation of $st$-orientations with a specific property, namely we address the following problem: ``Given an undirected graph $G$ and two prescribed vertices $s$ and $t$ for which $G \cup (s,t)$ is biconnected, compute an $st$-orientation of $G$ such that the resulting $st$-graph $G'$ has the minimum number of transitive edges (possibly none)''. We recall that an edge $(u,v)$ of a digraph $G'$ is \emph{transitive} if there exists a directed path from $u$ to $v$ in $G' \setminus (u,v)$. An $st$-orientation is \emph{non-transitive} if the resulting digraph has no transitive edges; $st$-graphs with no transitive edges are also known as \emph{transitively reduced $st$-graphs}~\cite{DBLP:books/ph/BattistaETT99,DBLP:journals/jgaa/EppsteinS13}, \emph{bipolar posets}~\cite{DBLP:journals/corr/abs-2105-06955}, or \emph{Hasse diagrams of lattices}~\cite{p-plpg-76,DBLP:journals/tcs/BattistaT88}. The problem we study, besides being of theoretical interest, has several practical motivations in graph drawing. We mention some of them: 
\begin{itemize}

	\item Planar $st$-oriented graphs without transitive edges admit compact dominance drawings with straight-line edges, a type of upward drawings that can be computed in linear time with very simple algorithms~\cite{DBLP:journals/dcg/BattistaTT92}; when a transitive edge is present, one can temporarily subdivide it with a dummy vertex, which will correspond to an edge bend in the final layout. Hence, having few transitive edges helps to reduce bends in a dominance drawing. 
	
	\item As previously mentioned, many layout algorithms for undirected planar graphs rely on a preliminary computation of an $st$-orientation of the input graph. We preliminary observed that reducing the number of transitive edges in such an orientation has typically a positive impact on the readability of the layout. Indeed, transitive edges often result in long curves; avoiding them produces faces where the lengths of the left and right paths are more balanced and leads to more compact drawings (see~\cref{fi:intro}). 
	%hence avoiding them helps to reduce edge length and to get more compact drawings (see~\cref{fi:intro}).
	
%	\item There are efficient algorithms to compute two-page book embeddings of two-terminal series-parallel digraphs, which either assume the absence of transitive edges~\cite{DBLP:conf/gd/AlzohairiR96} or which are easier to implement if transitive edges are not present~\cite{DBLP:journals/algorithmica/GiacomoDLW06}. These algorithms can be used to compute two-page book embeddings of (biconnected) undirected series-parallel graphs after a suitable $st$-orientation. 
	
	\item Algorithms for computing upward confluent drawings of transitively reduced DAGs are studied in~\cite{DBLP:journals/jgaa/EppsteinS13}. Confluent drawings exploit edge bundling to create ``planar'' layouts of non-planar graphs, without introducing ambiguity~\cite{DBLP:journals/jgaa/DickersonEGM05}. These algorithms can be applied to draw undirected graphs that have been previously $st$-oriented without transitive edges when possible.  
\end{itemize}

We also mention algorithms that compute two-page book embeddings of two-terminal series-parallel digraphs, which either assume the absence of transitive edges~\cite{DBLP:conf/gd/AlzohairiR96} or which are easier to implement if transitive edges are not present~\cite{DBLP:journals/algorithmica/GiacomoDLW06}.  

\paragraph{Contribution.} In this paper we first prove that deciding whether a graph admits an $st$-orientation without transitive edges is NP-complete. This is in contrast with the tractability of a problem that is at the opposite of ours, namely, deciding whether an undirected graph has an orientation such that the resulting digraph is its own transitive closure; this problem can be solved in linear time~\cite{DBLP:journals/dm/McConnellS99}.

From a practical point of view, we provide an Integer Linear Programming (ILP) model for planar graphs, whose solution is an $st$-orientation with the minimum number of transitive edges. In our setting, $s$ and $t$ are two prescribed vertices that belong to the same face of the input graph in at least one of its planar embeddings. We prove that the ILP model works very fast in practice. Popular solvers such as CPLEX can find a solution in few seconds for graphs up to $1000$ vertices and the resulting $st$-orientations save on average $35\%$ of transitive edges (with improvements larger than $80\%$ on some instances) with respect to applying classical unconstrained $st$-orientation algorithms. 
Moreover, focusing on popular graph drawing algorithms that apply an $st$-orientation as a preliminary step, we show that reducing the number of transitive edges leads to drawings that are much more compact.
%We also show that the $st$-orientations deriving from our ILP model can be used to compute more compact drawings with some classical polyline drawing algorithm.

For space restrictions, some details are omitted. Full proofs and additional material can be found in~\cref{se:app}.

\section{NP-Completeness of the General Problem}\label{se:hardness}

We prove that given an undirected graph $G=(V,E)$ and two vertices $s, t \in V$, it is NP-complete to decide whether there exists a non-transitive $st$-orientation of~$G$. We call this problem {\sc Non-Transitive st-Orientation (NTO)}. 
To prove the hardness of \textsc{NTO} we describe a reduction from the NP-complete problem {\sc Not-All-Equal 3SAT (NAE3SAT)}~\cite{s-csp-78}, where one has a collection of clauses, each composed of three literals out of a set $X$ of Boolean variables, and is asked to determine whether there exists a truth assignment to the variables in $X$ so that each clause has at least one \texttt{true} and one \texttt{false} literal. 

Starting from a \textsc{NAE3SAT} instance $\varphi$, we construct an instance $I_\varphi = \langle G, s,t \rangle$ of \textsc{NTO} such that $I_\varphi$ is a yes instance of \textsc{NAE3SAT} if and only if $\varphi$ is a yes instance of \textsc{NTO}.
Instance $I_\varphi$ has one variable gadget $V_x$ for each Boolean variable $x$ and one clause gadget $C_c$ for each clause $c$ of $\varphi$. By means of a split gadget, the truth value encoded by each variable gadget $V_x$ is transferred to all the clause gadgets containing either the direct literal $x$ or its negation $\overline{x}$.
Observe that the \textsc{NAE3SAT} instance is in general not ``planar'', in the sense that if you construct a graph where each variable $x$ and each clause $c$ is a vertex and there is an edge between $x$ and $c$ if and only if a literal of $x$ belongs to~$c$, then such a graph would be non-planar. The \textsc{NAE3SAT} problem on planar instances is, in fact, polynomial~\cite{m-pnp-88}. Hence, $G$ has to be assumed non-planar as well.

% LONG: Before describing the gadgets, we introduce two simple observations on the constraints imposed by any non-transitive st-orientation of a graph $G$.

The main ingredient of the reduction is the \emph{fork gadget} (\cref{fi:gadget}), for which the following lemma holds (the proof is in \cref{sse:app-hardness}).

\begin{figure}[b]
	\centering
	\subfigure[]{
		\includegraphics[width=0.25\textwidth, page=1]{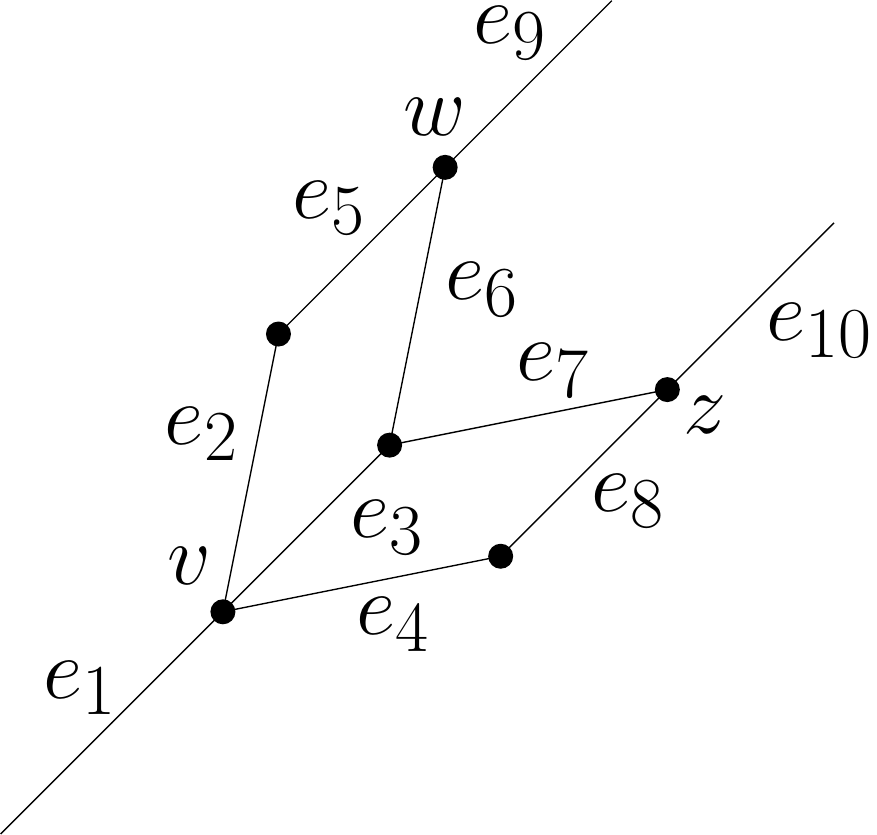}
		\label{fi:gadget-a}
	}
	\hfil
	\subfigure[]{
		\includegraphics[width=0.25\textwidth, page=2]{gadget}
		\label{fi:gadget-b}
	}
	\hfil
	\subfigure[]{
		\includegraphics[width=0.25\textwidth, page=3]{gadget}
		\label{fi:gadget-c}
	}
	\caption{(a) The fork gadget. (b)-(c) The two possible orientations of the fork gadget in a non-transitive st-orientation of the whole graph.}\label{fi:gadget} 
\end{figure}

\begin{restatable}{lemma}{leForkGadget}\label{le:fork-gadget}
Let $G$ be an undirected graph containing a fork gadget $F$ that does not contain the vertices $s$ or $t$. In any non-transitive st-orientation of $G$, the edges $e_9$ and $e_{10}$ of $F$ are oriented either both exiting $F$ or both entering $F$. They are oriented exiting $F$ if and only if edge $e_1$ is oriented entering $F$.  
\end{restatable}

\begin{figure}[tb]
	\centering
	\subfigure[]{
		\includegraphics[width=0.3\textwidth, page=1]{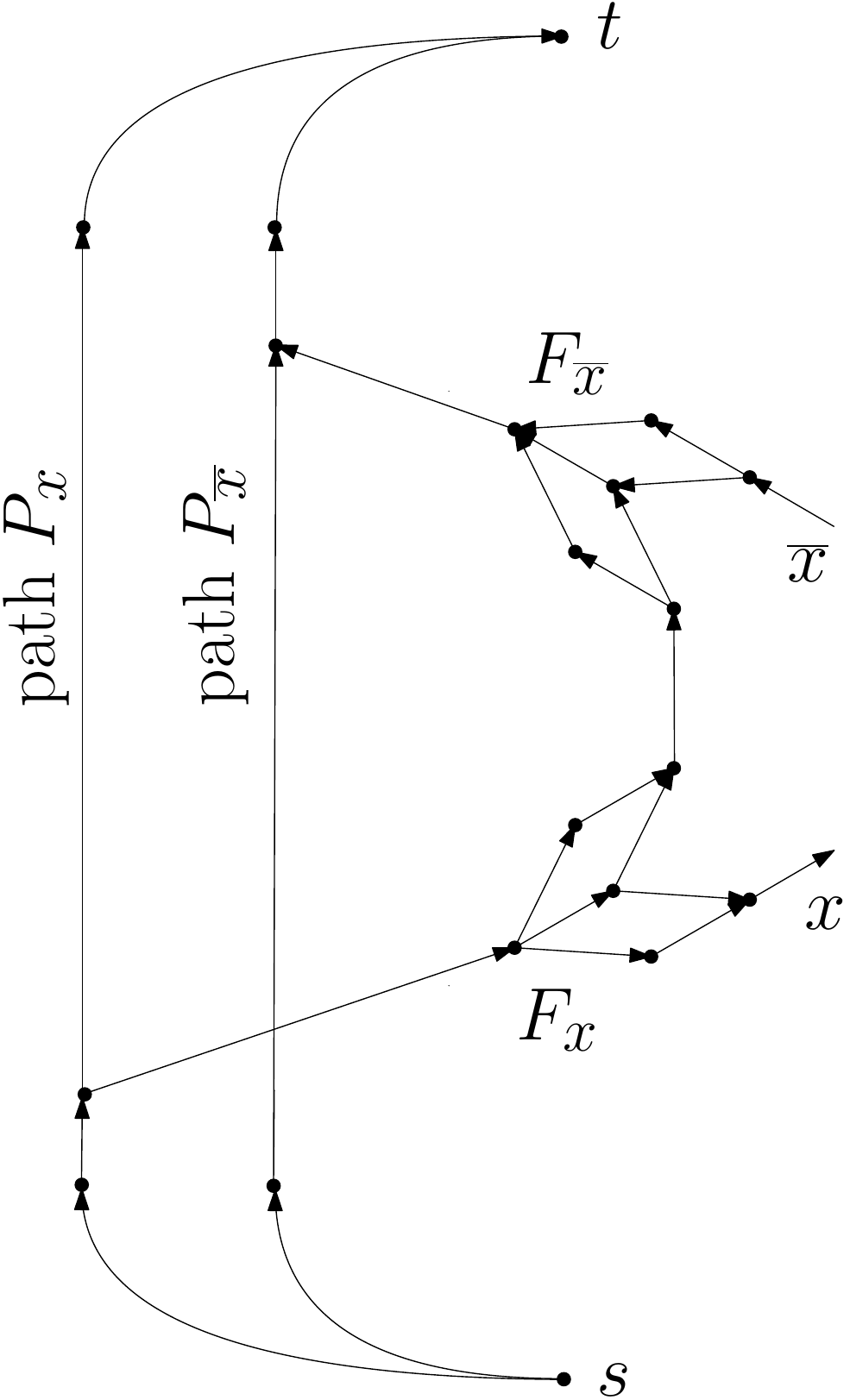}
		\label{fi:variable-gadget-a}
	}
	\hfil
	\subfigure[]{
		\includegraphics[width=0.3\textwidth, page=2]{variable-gadget}
		\label{fi:variable-gadget-b}
	}
	\caption{The variable gadget $V_x$ and its \texttt{true} (a) and \texttt{false} (b) orientations.}\label{fi:variable-gadget} 
\end{figure}

For each Boolean variable $x$ of $\phi$ we construct a \emph{variable gadget} $V_x$ by suitably combining two fork gadgets, denoted $F_x$ and $F_{\overline{x}}$, as follows (see \cref{fi:variable-gadget}). 
%(see \cref{fi:variable-gadget} in \cref{sse:app-hardness}). 
We introduce two paths $P_{x}$ and $P_{\overline{x}}$ of length four from $s$ to $t$. The edge $e_1$ of $F_x$ (of $F_{\overline{x}}$, respectively) is attached to the middle vertex of path $P_{x}$ (of path $P_{\overline{x}}$, respectively). Edge $e_{10}$ of $F_{\overline{x}}$ is identified with edge $e_9$ of $F_{x}$. The two edges $e_9$ of $F_{\overline{x}}$ and $e_{10}$ of $F_{x}$ are denoted $\overline{x}$ and $x$, respectively. We have the following lemma (see \cref{sse:app-hardness} for the proof).

\begin{restatable}{lemma}{leVariableGadget}\label{le:variable-gadget}
Let $G$ be an undirected graph containing a variable gadget $V_x$. In any non-transitive st-orientation of $G$ the two edges of $V_x$ denoted $x$ and $\overline{x}$ are one entering and one exiting $V_x$ or vice versa.
\end{restatable}

By virtue of \cref{le:variable-gadget} we associate the \texttt{true} value of variable $x$ with the orientation of $V_x$ where edge $x$ is oriented exiting and edge $\overline{x}$ is oriented entering $V_x$ 
(see \cref{fi:variable-gadget-a}). %in \cref{sse:app-hardness}).
We call such an orientation the \emph{\texttt{true} orientation of $V_x$}. Analogously, we associate the \texttt{false} value of variable $x$ with the orientation of $V_x$ where edge $x$ is oriented entering and edge $\overline{x}$ is oriented exiting $V_x$ (see \cref{fi:variable-gadget-b}). 
%in \cref{sse:app-hardness}).
%
Observe that edge $x$ (edge $\overline{x}$, respectively) is oriented exiting $V_x$ when the literal $x$ (the literal $\overline{x}$, respectively) is \texttt{true}. Otherwise edge $x$ (edge $\overline{x}$, respectively) is oriented entering $V_x$.

The \emph{split gadget} $S_k$ is composed of a chain of $k-1$ fork gadgets $F_1, F_2, \dots F_{k-1}$, where, for $i=1,2,\dots,k-2$, the edge $e_9$ of $F_i$ is identified with the edge $e_1$ of~$F_{i+1}$. We call \emph{input edge of $S_k$} the edge denoted $e_1$ of $F_1$.
Also, we call \emph{output edges of $S_k$} the $k-1$ edges denoted $e_{10}$ of the fork gadgets $F_1, F_2, \dots F_{k-1}$ and the edge $e_9$ of $F_{k-1}$ (see \cref{fi:split-gadget}). The next lemma is immediate and we~omit~the~proof.

\begin{figure}[tb]
	\centering
		\includegraphics[width=0.8\textwidth, page=1]{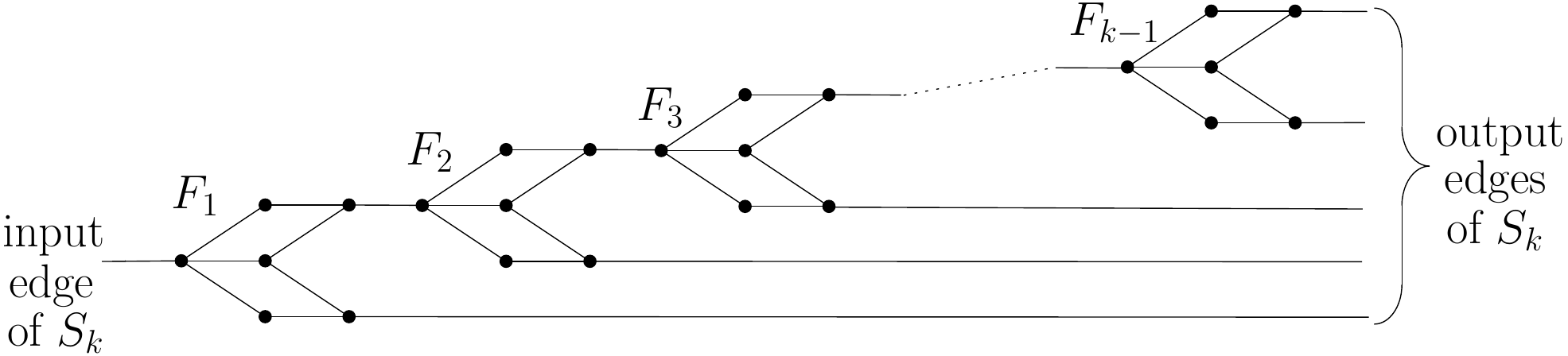}
	\caption{The split gadget $S_k$.}\label{fi:split-gadget} 
\end{figure}

\begin{lemma}\label{le:split-gadget}
Let $G$ be an undirected graph containing a split gadget $S_k$ that does not contain the vertices $s$ or $t$. In any non-transitive st-orientation of $G$, the $k$ output edges of $S_k$ are all oriented exiting $S_k$ if the input edge of $S_k$ is oriented entering $S_k$. Otherwise, if the input edge of $S_k$ is oriented exiting $S_k$ the ouput edges of $S_k$ are all oriented entering $S_k$.
\end{lemma}

\begin{figure}[tb]
	\centering
		\includegraphics[width=0.7\textwidth, page=1]{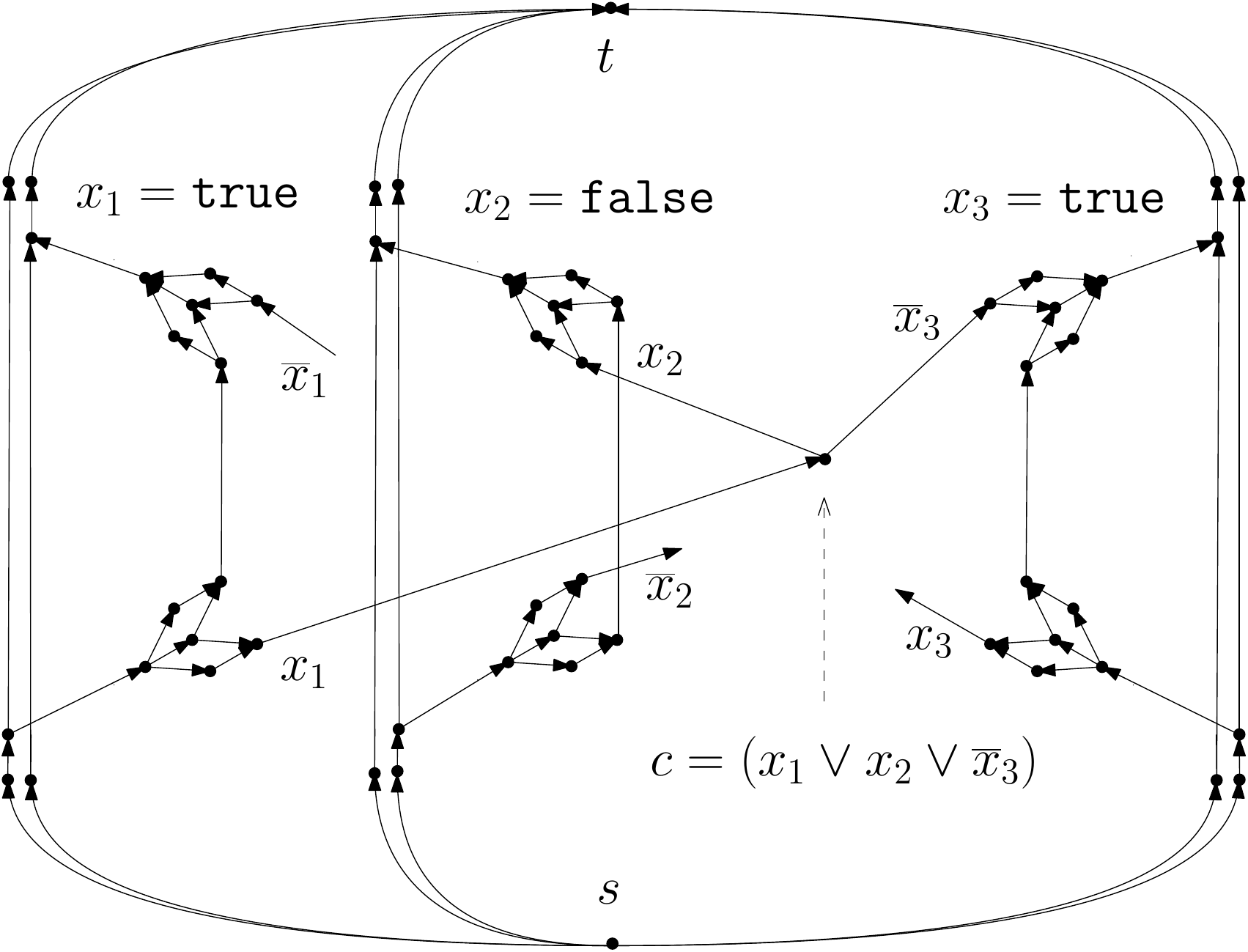}
	\caption{The clause gadget $C_c$ for clause $c = (x_1 \vee x_2 \vee \overline{x}_3)$. The configurations of the three variable gadgets correspond to the truth values $x_1 = \texttt{true}$, $x_2  = \texttt{false}$, and $x_3 = \texttt{true}$. The clause is satisfied because the first literal $x$ is \texttt{true} and the second and third literals $x_2$ and $\overline{x}_3$ are \texttt{false}.}\label{fi:clause-gadget} 
\end{figure}

If the directed literal $x$ (negated literal $\overline{x}$, respectively) occurs in $k$ clauses, we attach the edge denoted $x$ (denoted $\overline{x}$, respectively) of $V_x$ to a split gadget $S_x$, and use the $k$ output edges of $S_x$ to carry the truth value of $x$ (of $\overline{x}$, respectively) to the $k$ clauses.  
The \emph{clause gadget} $C_c$ for a clause $c = (l_1 \vee l_2 \vee l_3)$ is simply a vertex  $v_c$ that is incident to three edges encoding the truth values of the three literals $l_1$, $l_2$, and $l_3$ (see \cref{fi:clause-gadget}). We prove the following.

%\begin{restatable}{theorem}{thHardness}\label{th:hardness}
%Instance $I_\varphi$ is a positive instance of \textsc{NTO} if and only if instance $\varphi$ is a positive instance of \textsc{NAE3SAT}.
%\end{restatable}

\begin{restatable}{theorem}{thHardness}\label{th:hardness}
\textsc{NTO} is NP-complete.
\end{restatable}
\begin{sketch}
The reduction from an instance $\varphi$ of \textsc{NAE3SAT} to an instance~$I_\varphi$ described above is performed in time linear in the size of $\varphi$. Also, $I_\varphi$ is positive if and only if $\varphi$ is positive. Indeed, in any non-transitive $st$-orientation of $G$ each vertex $v_c$ of a clause gadget $C_c$ has at least one incoming and one outgoing edge, as well as in any truth assignment that satisfies $\varphi$ each clause $c$ has at least one \texttt{true} and one \texttt{false}~literal. Finally, \textsc{NTO} is trivially in NP, as one can non-deterministically explore all possible orientations of the~graph. 
\end{sketch}

%Since the reduction described above can be performed in time linear in the size of the instance $\varphi$ of \textsc{NAE3SAT}, problem \textsc{NTO} is NP-hard.
% LONG: Observe that the analogous problem where the source and the target vertices of $G$ are not prescribed but can be freely choosen is also NP-hard. Problem \textsc{NTO}, in fact, can be easily reduced to it. Consider an instance $\langle G^*, s^*, t^* \rangle$ of \textsc{NTO}. Add two vertices $s^+$ and $t^+$ to $G^*$ and connect them to $s^*$ and to $t^*$, respectively. Call $G^+$ the obtained graph. Since $s^+$ and $t^+$ have degree one in $G^+$, in any non-transitive st-orientation of $G^+$ they can only be sources or sinks, where if one of them is the source the other one is the sink. Hence, given any non-transitive st-orientation of $G^+$ you can immediately find a non-transitive $s^*t^*$-orientation of $G^*$, possibly by reversing all edge orientations if $t^+$ is the source and $s^+$ is the sink. Conversely, given a non-transitive $s^*t^*$-orientation of $G^*$ you easily find an st-orientation of $G$ orienting the edge $(s^+,s^*)$ from $s^+$ to $s^*$ and the edge $(t^*,t^+)$ from $t^*$ to $t^+$.
The analogous problem where the source and the target vertices of $G$ are not prescribed but can be freely choosen is also NP-complete (see \cref{sse:app-hardness}).

\section{ILP Model for Planar Graphs}\label{se:ilp}
%\textcolor{red}{Carla+Walter}
%Describe the ILP model for planar graphs, through which we are able to find a planar $st$-orientation of a planar biconnected graph ($s$ and $t$ are given) with the minimum number of transitive edges.  It relies on a characterization of the planar $st$-orientations of plane graphs without transitive edges. Notice that any planar embedding can be chosen to support the characterization (and therefore the model), but of course the result is not restricted to that embedding (because the transitivity of an edge does not depend on the specific planar embedding).

%Observe that the ILP model uses a number of variables that is linear in the number of vertices of the graph (we expect that it behaves very quickly in practice).

Let $G$ be a planar graph with two prescribed vertices $s$ and $t$, such that $G \cup (s,t)$ is biconnected and such that $G$ admits a planar embedding with $s$ and $t$ on the external face. In this section we describe how to compute an $st$-orientation of $G$ with the minimum number of transitive edges by solving an ILP model.    

Suppose that $G'$ is the plane $st$-graph resulting from a planar $st$-orientation of $G$, along with a planar embedding where $s$ and $t$ are on the external face. It is well known (see, e.g.,~\cite{DBLP:books/ph/BattistaETT99}) that 
for each vertex $v \neq s,t$ in $G'$, all incoming edges of $v$ (as well as all outgoing edges of $v$) appear consecutively around $v$. Thus, the circular list of edges incident to $v$ can be partitioned into two linear lists, one containing the incoming edges of $v$ and the other containing the outgoing edges of $v$. Also, the boundary of each internal face $f$ of $G'$ consists of two edge-disjoint directed paths, called the \emph{left path} and the \emph{right path} of $f$, sharing the same end-vertices (i.e., the same source and the same destination). 
%If $f$ is the external face, then $u$ coincides with $s$ and $v$ coincides with $t$. 
It can be easily verified that an edge $e$ of $G'$ is transitive if and only if it coincides with either the left path or the right path of some face of $G'$ (see also Claim 2 in~\cite{DBLP:journals/corr/abs-2105-06955}). Note that, since the transitivity of $e$ does not depend on the specific planar embedding of~$G'$, the aforementioned property for $e$ holds for every planar embedding of $G'$. Due to this observation, in order to compute a planar $st$-orientation of $G$ with the minimum number of transitive edges, we can focus on any arbitrarily chosen planar embedding of $G$ with $s$ and $t$ on the external face.   

Let $e_1$ and $e_2$ be two consecutive edges encountered moving clockwise along the boundary of a face $f$, and let $v$ be the vertex of $f$ shared by $e_1$ and $e_2$. The triple $(e_1,v,e_2)$ is an \emph{angle of $G$ at $v$ in $f$}.
Denote by $\deg(f)$ the number of angles in $f$ and by $\deg(v)$ the number of angles at $v$.
As it was proved in~\cite{DBLP:journals/jgaa/DidimoP03}, all planar $st$-orientations of the plane graph $G$ can be characterized in terms of labelings of the angles of $G$. Namely, each planar $st$-orientation of $G$ has a one-to-one correspondence with an angle labeling, called an \emph{$st$-labeling} of~$G$, that satisfies the following properties: 
\begin{itemize}
\item[\textsf{(L1)}] Each angle is labeled either S (small) or F (flat), except the angles at $s$ and at $t$ in the external face, which are not labeled; 
\item[\textsf{(L2)}] Each internal face $f$ has 2 angles labeled S and $\deg(f)-2$ angles labeled~F; 
\item[\textsf{(L3)}] For each vertex $v \neq s,t$ there are $\deg(v)-2$ angles at $v$ labeled~S and $2$ angles at $v$ labeled~F; 
\item[\textsf{(L4)}] All angles at $s$ and $t$ in their incident internal faces are labeled S.
\end{itemize}

Given an $st$-labeling of $G$, the corresponding $st$-orientation of $G$ is such that for each vertex $v \neq s,t$, the two F angles at $v$ separate the list of incoming edges of $v$ to the list of outgoing edges of $v$, while the two S angles in a face $f$ separate the left and the right path of $f$. See Fig.~\ref{fi:st-labeling} for an illustration. The $st$-orientation can be constructed from the $st$-labeling in linear time by a breadth-first-search of~$G$ that starts from~$s$, makes all edges of $s$ outgoing, and progressively orients the remaining edges of~$G$ according to the angle labels.

\begin{figure}[tb]
	\centering
	\subfigure[]{
		\includegraphics[height=0.45\textwidth, page=1]{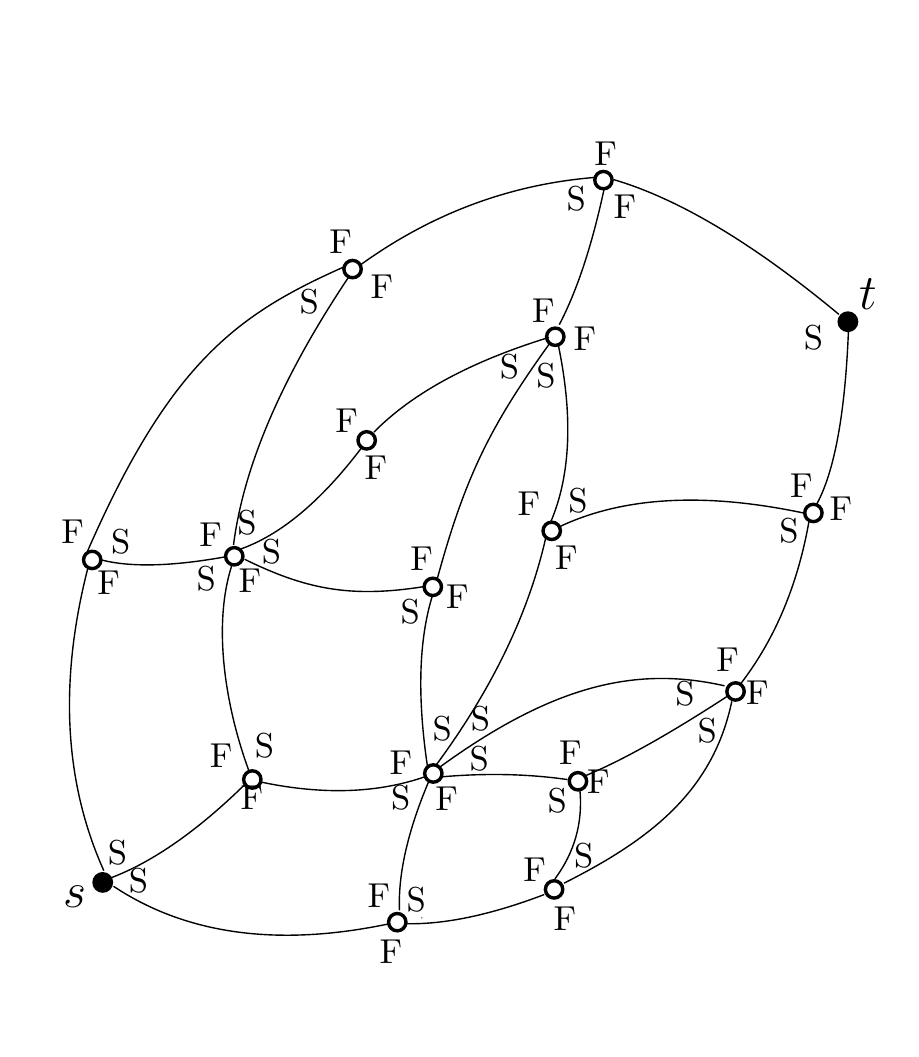}
		\label{fi:st-labeling-a}
	}
	\hfil
	\subfigure[]{
		\includegraphics[height=0.45\textwidth, page=3]{st-labeling}
		\label{fi:st-labeling-b}
	}
	\caption{(a) An $st$-labeling of a plane graph $G$ with prescribed nodes $s$ and $t$. (b) The corresponding $st$-orientation of $G$. \label{fi:st-labeling}}
\end{figure}

Thanks to the characterization above, an edge $e=(u,v)$ of the $st$-graph resulting from an $st$-orientation is transitive if and only if in the corresponding $st$-labeling the angle at $u$ and the angle at $v$ in one of the two faces incident to $e$ (possibly in both faces) are labeled S. Based on this, we present an ILP model that describes the possible $st$-labelings of $G$ (for any arbitrary planar embedding of $G$ with $s$ and $t$ on the external face) and that minimizes the number of transitive edges. The model aims to assign angle labels that satisfy Properties \textsf{(L1)--(L4)} and counts pairs of consecutive S labels that occur in the circular list of angles in an internal face; additional constraints are needed to avoid that a transitive edge is counted twice when it coincides with both the left and the right path of its two incident faces. The model, which uses a number of variables and constraints that is linear in the size of $G$, is as follows.

\bigskip\noindent{\bf Sets.} Denote by $V$, $E$, and $F$ the sets of vertices, edges, and faces of $G$, respectively. Also let $F_{\rm int} \subset F$ be the set of internal faces of $G$. For each face $f \in F$, let $V(f)$ and $E(f)$ be the set of vertices and the set of edges incident to $f$, respectively. For each vertex $v \in V$, let $F(v)$ be the set of faces incident to $v$ and let $F_{\rm int}(v)$ be the set of internal faces incident to $v$. For each edge $e \in E$, let $F(e)$ be the set consisting of the two faces incident to~$e$.

\smallskip\noindent{\bf Variables.} We define a binary variable $x_{vf}$ for each  vertex $v \in V \setminus \{s,t\}$ and for each face $f \in F(v)$. Also, we define the binary variables $x_{sf}$ (resp. $x_{tf}$) for each face $f \in F_{\rm int}(s)$ (resp. $f \in F_{\rm int}(t)$). If $x_{vf}=1$ (resp. $x_{vf}=0$)  we assign an S label (resp. an F label) to the angle at $v$ in $f$.
 
For each internal face $f \in F_{\rm int}$ and for each edge $(u,v) \in E(f)$, we define a binary variable $y_{uvf}$. An assignment $y_{uvf}=1$ indicates that both the angles at $u$ and at $v$ in $f$ are labeled S, that is, $x_{uf}=1$ and $x_{vf}=1$. As a consequence, if $y_{uvf}=1$ edge $(u,v)$ is transitive. Note however that the sum of all $y_{uvf}$ does not always correspond to the number of transitive edges; indeed, if $f$ and $g$ are the two internal faces incident to edge $(u,v)$, it may happen that both $y_{uvf}$ and $y_{uvg}$ are set to one, thus counting $(u,v)$ as transitive twice. To count the number of transitive edges without repetitions, we introduce another binary variable $z_{uv}$, for each edge $(u,v) \in E$, such that $z_{uv}=1$ if and only if $(u,v)$ is transitive. 

\smallskip\noindent{\bf Objective function and constraints.} The objective function and the set of constraints are described by the formulas~$(1)$--$(8)$. The objective is to minimize the total number of transitive edges, i.e., the sum of the variables~$z_{uv}$. 
%As we already explained, although we use a specific planar embedding for $G$ to determine an $st$-labeling, its corresponding $st$-orientation will represent the optimal solution over all $st$-planar embeddings of $G$. 
Constraints~\ref{ilp:S-in-face} and~\ref{ilp:S-at-vertex} guarantee Properties~\textsf{(L2)} and~\textsf{(L3)} of the $st$-labeling, respectively, while Constraints~\ref{ilp:S-at-s} and~\ref{ilp:S-at-t} guarantee Property~\textsf{(L4)}. Constraints~\ref{ilp:penalty-transitive-edge} relate the values of the variables $y_{uvf}$ to the values of $x_{uf}$ and $x_{vf}$. Namely, they guarantee that $y_{uvf}=1$ if and only if both $x_{uf}$ and $x_{vf}$ are set to 1. Constraints~\ref{ilp:avoid-double-penalty} relate the values of the variables $z_{uv}$ to those of the variables $y_{uvf}$; they guarantee that an edge $(u,v)$ is counted as transitive (i.e., $z_{uv}=1$) if and only if in at least one of the two faces $f$ incident to $(u,v)$ both the angle at $u$ and the angle at $v$ are labeled S. Finally, we explicitly require that $x_{uv}$ and $y_{uv}$ are binary variables, while we only require that each $z_{uv}$ is a non-negative integer; this helps to speed-up the solver and, along with the objective function, is enough to guarantee that each $z_{uv}$ takes value 0 or 1.

\begin{align}
\min \sum_{(u,v) \in E}z_{uv} \label{ilp:obj}\\
\sum_{v \in V(f)} x_{vf} = 2 \;\;\;\;\;\; \forall f \in F_{\rm int} \label{ilp:S-in-face}\\
\sum_{f \in F(v)} x_{vf} = \deg(v)-2  \;\;\;\;\;\; \forall v \in V \setminus \{s,t\} \label{ilp:S-at-vertex}\\
x_{sf}=1 \;\;\;\;\;\; \forall f \in F_{\rm int} \cap F(s) \label{ilp:S-at-s}\\
x_{tf}=1 \;\;\;\;\;\; \forall f \in F_{\rm int} \cap F(t) \label{ilp:S-at-t}\\
x_{uf} + x_{vf} \leq y_{uvf} + 1  \;\;\;\;\;\; \forall f \in F_{\rm int} \;\;\;\ \forall (u,v) \in E(f) \label{ilp:penalty-transitive-edge}\\
z_{uv} \geq y_{uvf} \;\;\;\;\;\; \forall e=(u,v) \in E \;\;\;\;\;\; \forall f \in F(e) \label{ilp:avoid-double-penalty}\\
x_{vf} \in \{0,1\} \;\;\;\; y_{uvf} \in \{0,1\} \;\;\;\; z_{uv} \in \mathbb{N} 	
\end{align}

%\medskip
%It is worth observing that the total number of variables and constraints of the ILP model is linear in the number of vertices of the graph. 

\section{Experimental Analysis}\label{se:experiments}

We evaluated the ILP model with the solver IBM ILOG CPLEX 20.1.0.0 (using the default setting), running on a laptop with Microsoft Windows 11 v.10.0.22000 OS, Intel Core i7-8750H 2.20GHz CPU, and 16GB RAM.

\smallskip\noindent{\bf Instances.} The experiments have been executed on a large benchmark of instances, each instance consisting of a plane biconnected graph and two vertices $s$ and $t$ on the external face. These graphs are randomly generated with the same approach used in previous experiments in graph drawing (see, e.g.,~\cite{DBLP:journals/tc/BertolazziBD00}). Namely, for a given integer $n>0$, we generate a plane graph with $n$ vertices starting from a triangle and executing a sequence of steps, each step preserving biconnectivity and planarity. At each step the procedure randomly performs one of the two following operations: $(i)$ an Insert-Edge operation, which splits a face by adding a new edge, or $(ii)$ an Insert-Vertex operation, which subdivides an existing edge with a new vertex. The Insert-Vertex operation is performed with a prescribed probability $p_{\rm iv}$ (which is a parameter of the generation process), while the Insert-Edge operation is performed with probability $1-p_{\rm iv}$. For each operation, the elements (faces, vertices, or edges) involved are randomly selected with uniform probability distribution. To avoid multiple edges, if an Insert-Edge operation selects two end-vertices that are already connected by an edge, we discard the selection and repeat the step. Once the plane graph is generated, we randomly select two vertices $s$ and $t$ on its external face, again with uniform probability distribution.
We generated  a sample of 10 instances for each pair $(n,p_{\rm iv})$, with $n \in \{10,20, \dots, 90, 100, 200, \dots, 900, 1000\}$ and $p_{\rm iv} \in \{0.2, 0.4, 0.5, 0.6, 0.8\}$, for a total of 950 graphs. Note that, higher values of $p_{\rm iv}$ lead to sparser graphs.

Table~\ref{ta:density} in the appendix reports for each sample the average, the minimum, and the maximum density (number of edges divided by the number of vertices) of the graphs in that sample, together with the standard deviation. On average, for $p_{\rm iv}=0.8$ we have graphs with density of $1.23$ (close to the density of a tree), for $p_{\rm iv}=0.5$ we have graphs with density of $1.76$, and for $p_{\rm iv}=0.2$ we have graphs with density $2.53$ (close to the density of maximal planar graphs).  

\medskip

\smallskip\noindent{\bf Experimental Goals.} We have three main experimental goals:
\textsf{(G1)} Evaluate the efficiency of our approach, i.e., the running time required by our ILP model; 
\textsf{(G2)} Evaluate the percentage of transitive edges in the solutions of the ILP model and how many transitive edges are saved w.r.t. applying a classical linear-time algorithm that computes an unconstrained $st$-orientation of the graph~\cite{DBLP:journals/tcs/EvenT77}; 
\textsf{(G3)} Evaluate the impact of minimizing the number of transitive edges on the area (i.e. the area of the minimum bounding box) of polyline drawings constructed with algorithms that compute an $st$-orientation as a preliminary step.

%\smallskip\noindent{\bf Experimental Goals.} We have three main experimental goals:
%\textsf{(G1)} Evaluate the efficiency of our approach, i.e., the running time required to find a solution of the ILP model; 
%\textsf{(G2)} Evaluate the percentage of transitive edges in the solutions of the ILP model and how many transitive edges are saved with respect to applying a classical linear-time algorithm that computes an unconstrained $st$-orientation of the graph~\cite{DBLP:journals/tcs/EvenT77}; 
%\textsf{(G3)} Evaluate the impact of minimizing the number of transitive edges on the area of polyline drawings constructed with algorithms that compute an $st$-orientation as a preliminary step. 

%\begin{itemize}
%	\item[\textsf{(G1)}] Evaluate the efficiency of our approach, i.e., the running time required to find a solution of the ILP model; 
%	\item[\textsf{(G2)}] Evaluate the percentage of transitive edges in the solutions of the ILP model and how many transitive edges are saved with respect to applying a classical linear-time algorithm that computes an unconstrained $st$-orientation of the graph~\cite{DBLP:journals/tcs/EvenT77}; 
%	\item[\textsf{(G3)}] Evaluate the impact of minimizing the number of transitive edges on the area of polyline drawings constructed with algorithms that compute an $st$-orientation as a preliminary step. 
%\end{itemize}
About~(G1), we refer to the algorithm that solves the ILP model as \textsc{OptST}. 
About~(G2) and~(G3) we used implementations available in the GDToolkit library~\cite{DBLP:reference/crc/BattistaD13} for the following algorithms:
%
%\begin{itemize}
	$(a)$ A linear-time algorithm that computes an unconstrained $st$-orientation of the graph based on the classical $st$-numbering algorithm by Even and Tarjan~\cite{DBLP:journals/tcs/EvenT77}. We refer to this algorithm as \textsc{HeurST}.     
	$(b)$ A linear-time algorithm that first computes a visibility representation of an
	undirected planar graph based on a given $st$-orientation of the graph, and then computes from this representation a planar polyline drawing~\cite{DBLP:journals/tcs/BattistaT88}. We call \textsc{DrawHeurST} and \textsc{DrawOptST} the applications of this drawing algorithm to the $st$-graphs obtained by \textsc{HeurST} and of \textsc{OptST}, respectively. 
%\end{itemize}

\begin{figure}[tb]
	\centering
	\subfigure[]{
		{\includegraphics[width=0.48\textwidth]{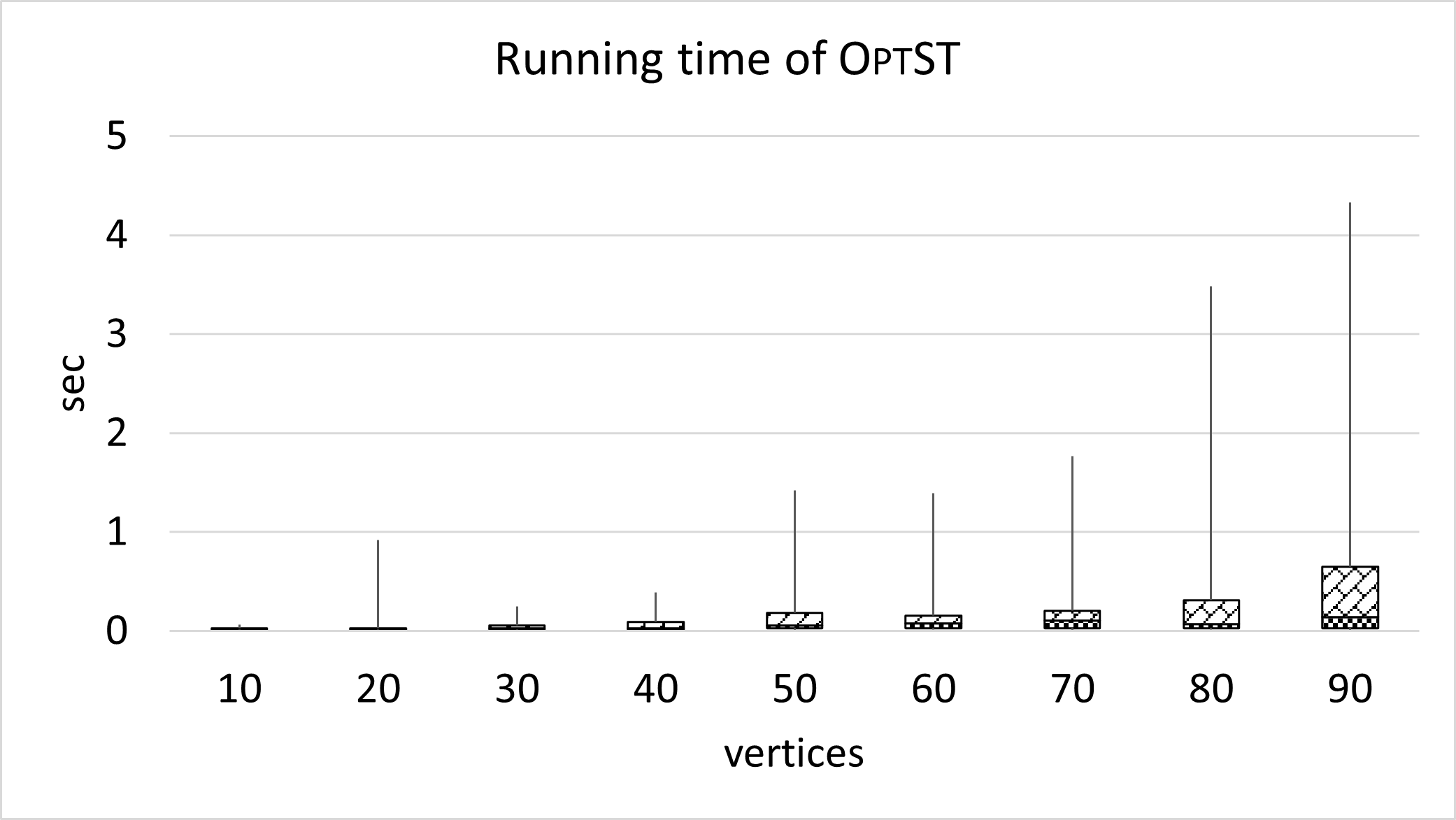}\label{fi:boxplot_10-90}}
	}\hfill
	\subfigure[]{
		{\includegraphics[width=0.48\textwidth]{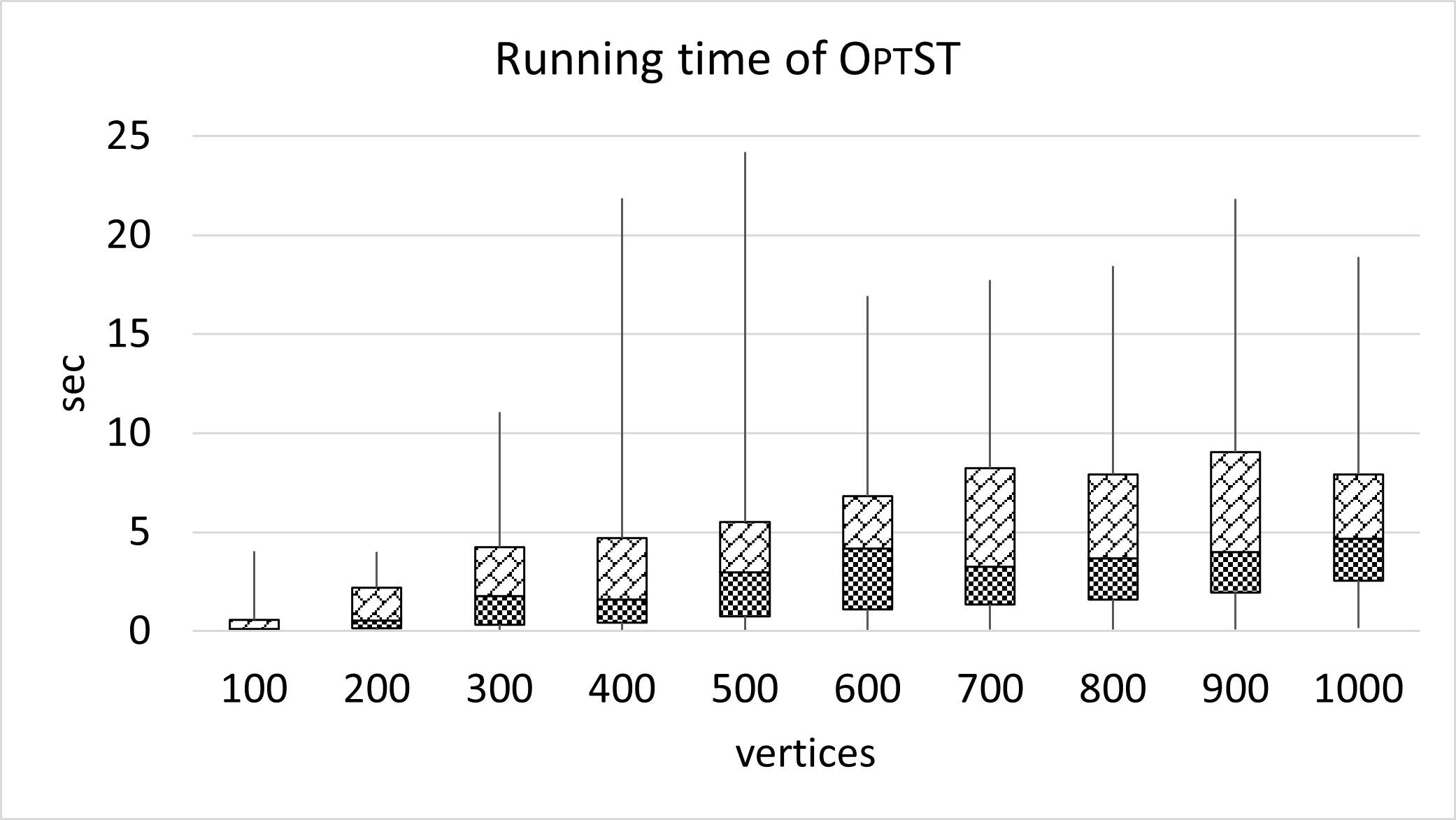}\label{fi:boxplot_100-1000}}
	}
	\caption{Box-plots of the running time of \textsc{OptST}.}\label{fi:box-plot}
\end{figure}
  
\smallskip\noindent{\bf Experimental Results.} About~(G1), \cref{fi:box-plot} reports the running time (in seconds) of \textsc{OptST}, i.e., the time needed by CPLEX to solve our ILP model. To make the charts more readable we split the results into two sets, one for the instances with number of vertices up to 90 and the other for the larger instances. \textsc{OptST} is rather fast: 75\% of the instances with up to 90 vertices is solved in less than one second and all these instances are solved in less than five seconds. For the larger instances (with up to 1000 vertices), 75\% of the instances are solved in less than 10 seconds and all instances are solved in less than 25 seconds. These results clearly indicate that our ILP model can be successfully used in several application contexts that manage graphs with up to thousand vertices.
%From an application perspective these results clearly indicate that our ILP model can be successfully used in practice.

% Transitive edges improvement

\begin{figure}[tb]
	\centering
	\subfigure[]{
		{\includegraphics[width=0.475\textwidth]{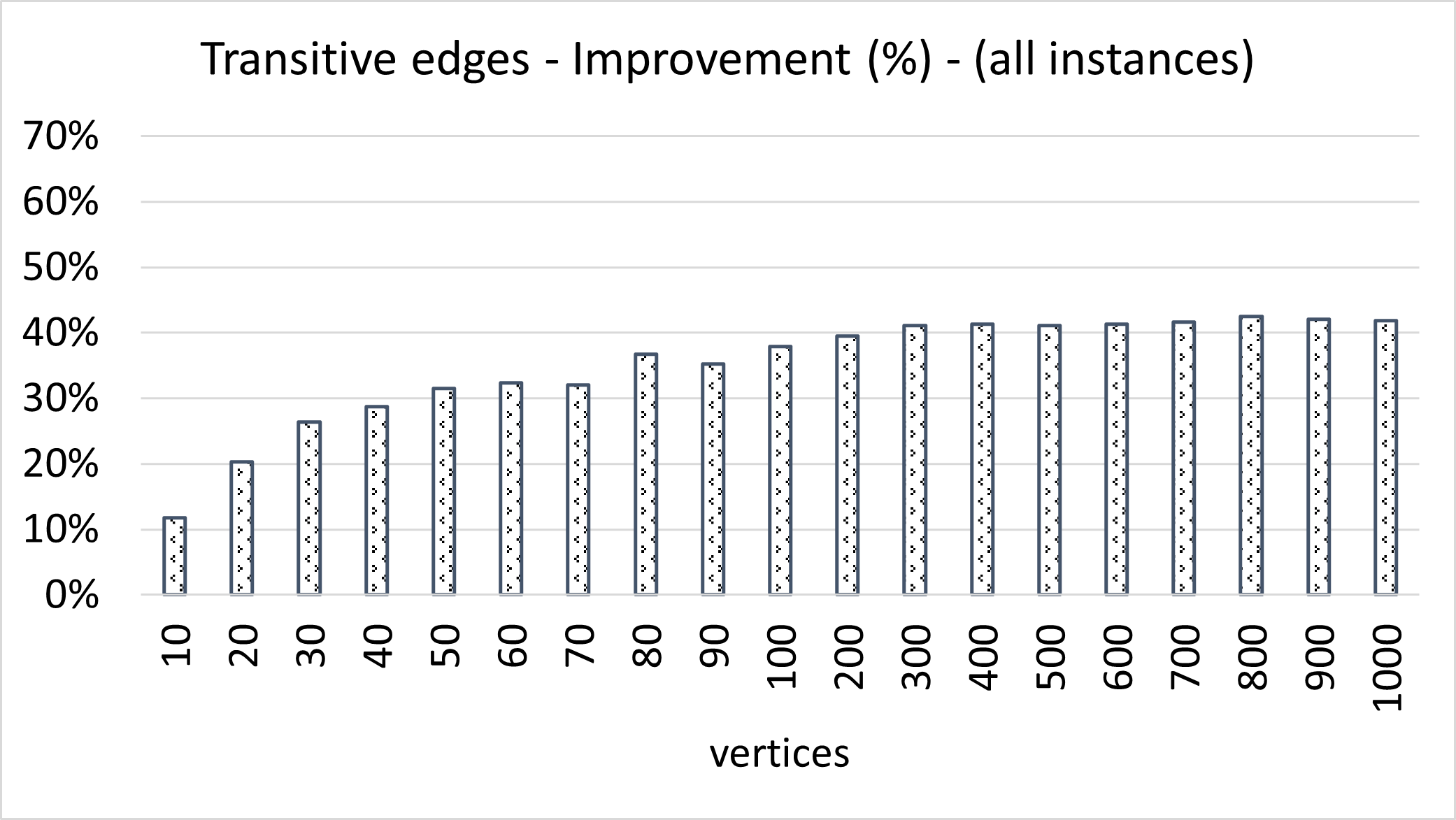}\label{fi:impTrE_all}}
	}
	\subfigure[]{
		{\includegraphics[width=0.475\textwidth]{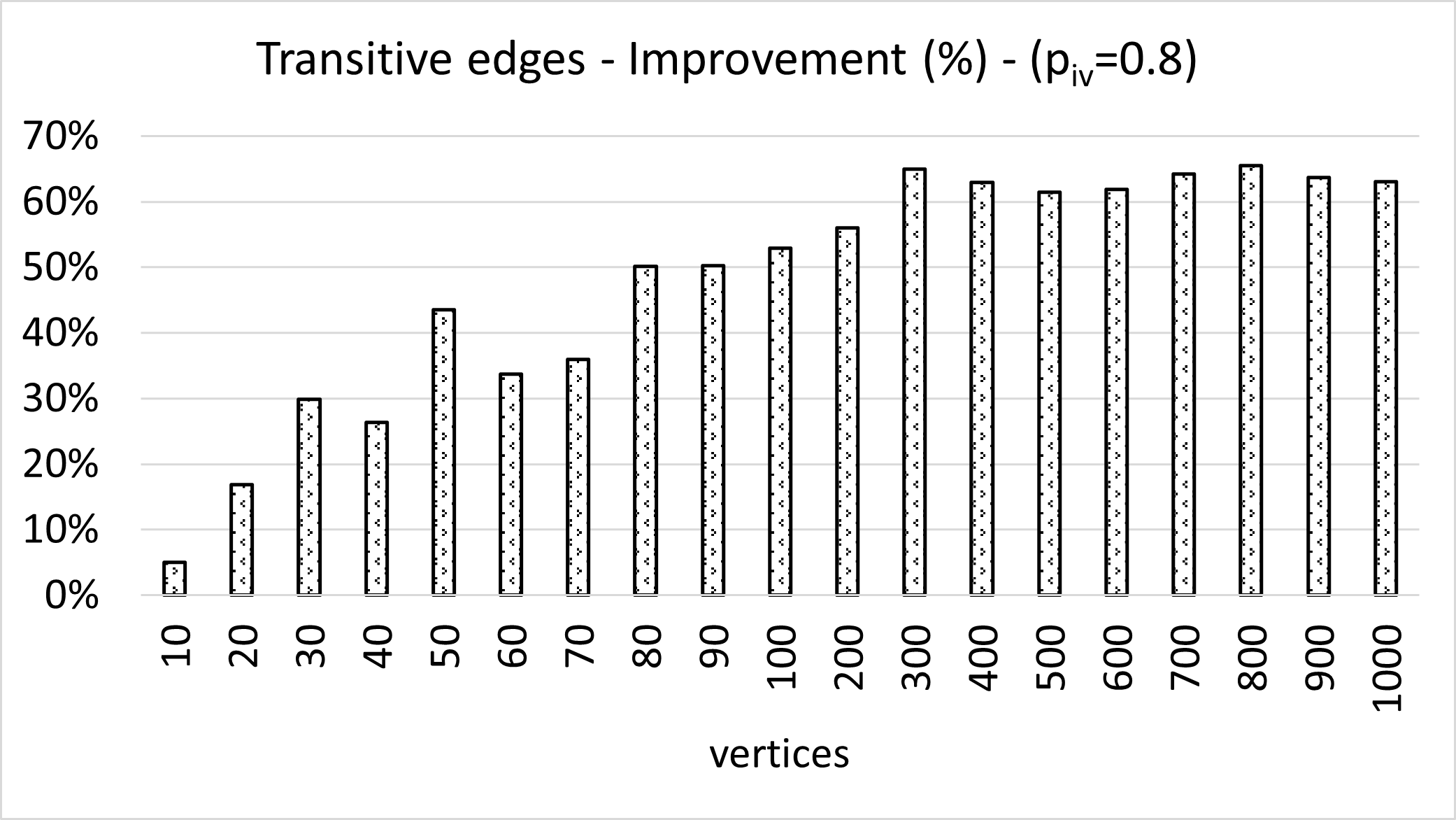}\label{fi:impTrE_0-8}}
	}\hfill
	\subfigure[]{
		{\includegraphics[width=0.475\textwidth]{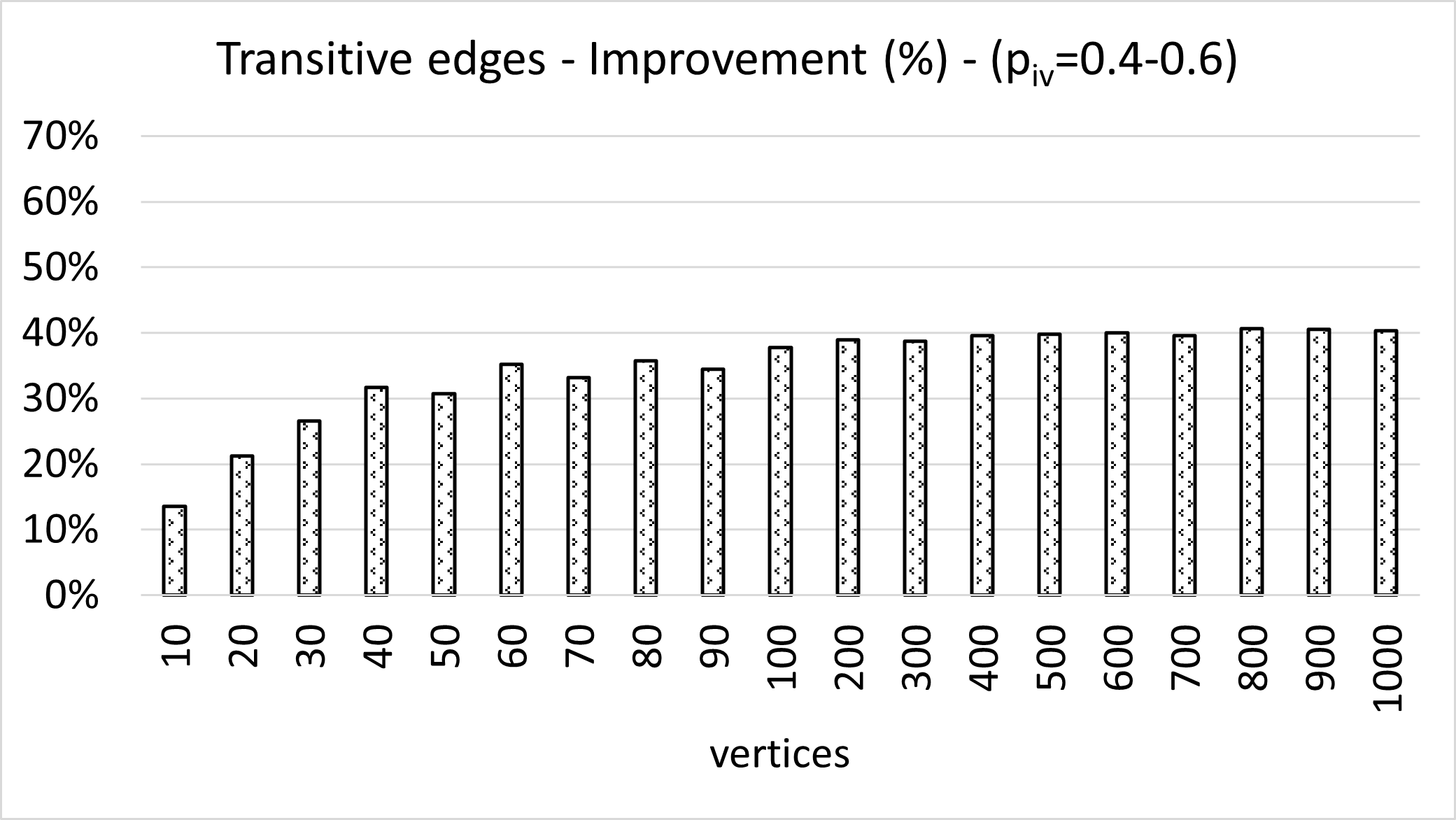}\label{fi:impTrE_0-6_0-4}}
	}\hfill
	\subfigure[]{
		{\includegraphics[width=0.475\textwidth]{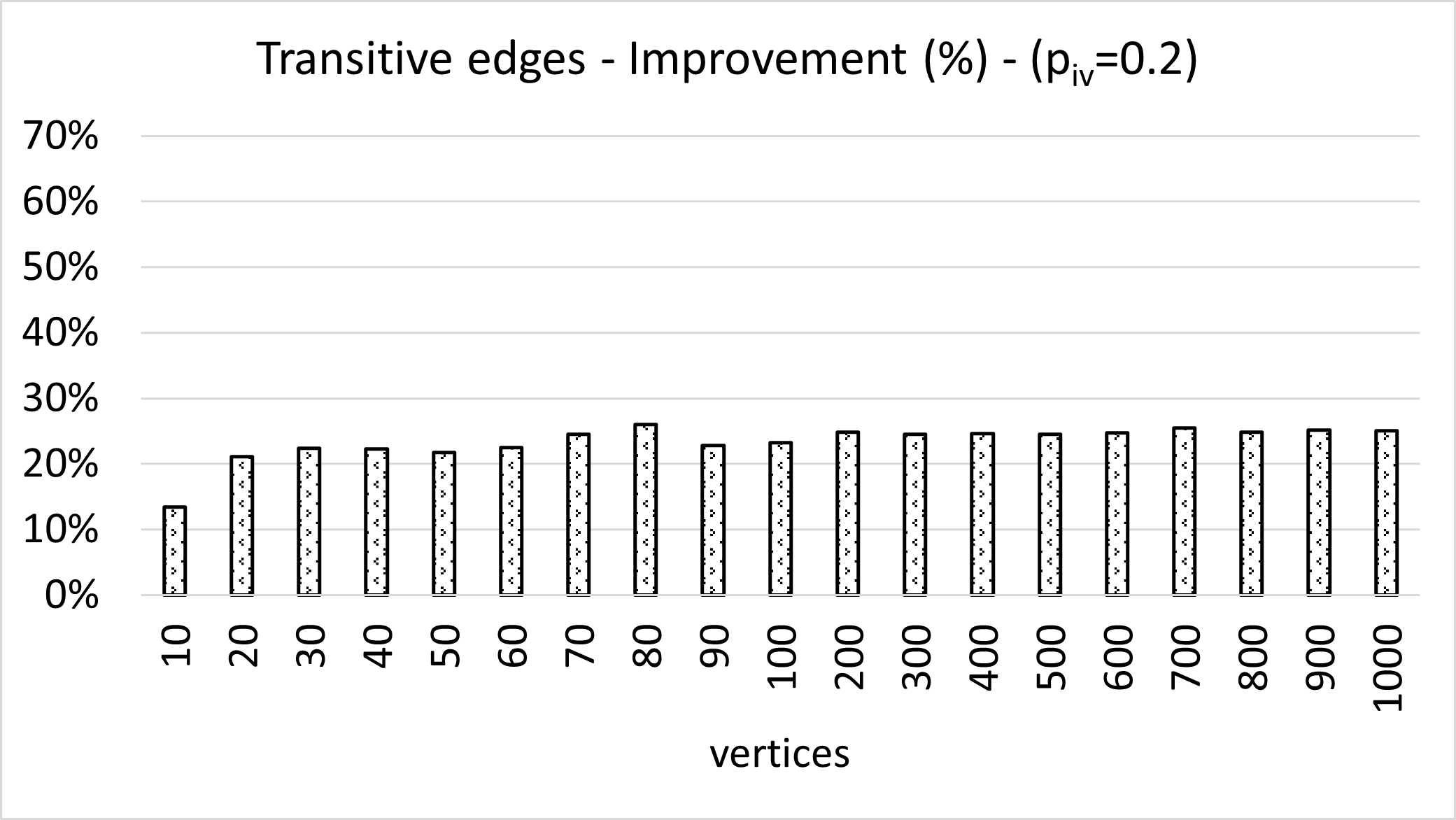}\label{fi:impTrE_0-2}}
	}\caption{Improvement (\%) in the number of transitive edges.}\label{fi:impTrE}
\end{figure}

% DrawOptST vs DrawOptHeur
\begin{figure}[tb]
	\centering
	\subfigure[]{
		{\includegraphics[width=0.475\textwidth]{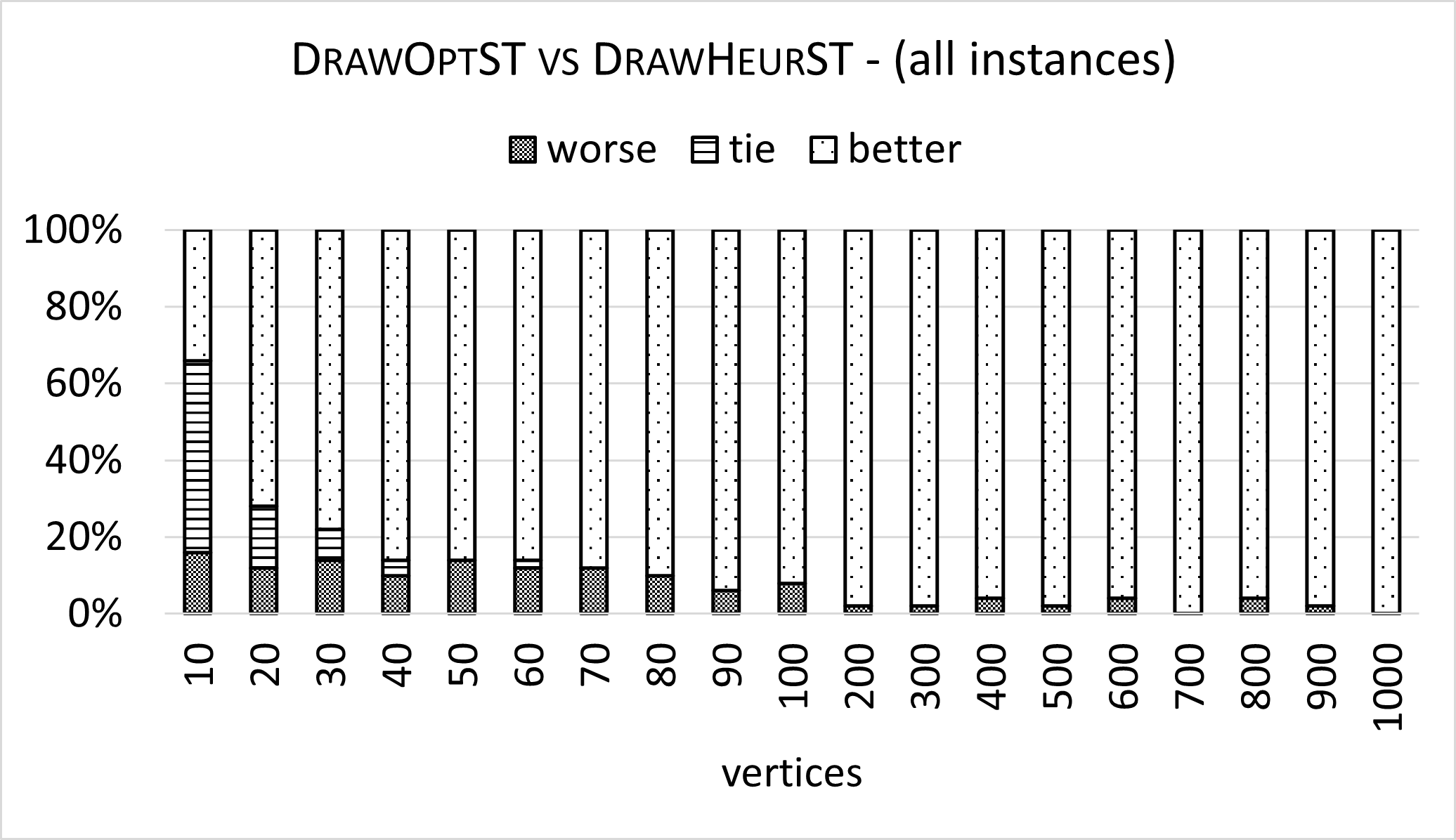}\label{fi:impInstances_all}}
	}
	\subfigure[]{
		{\includegraphics[width=0.475\textwidth]{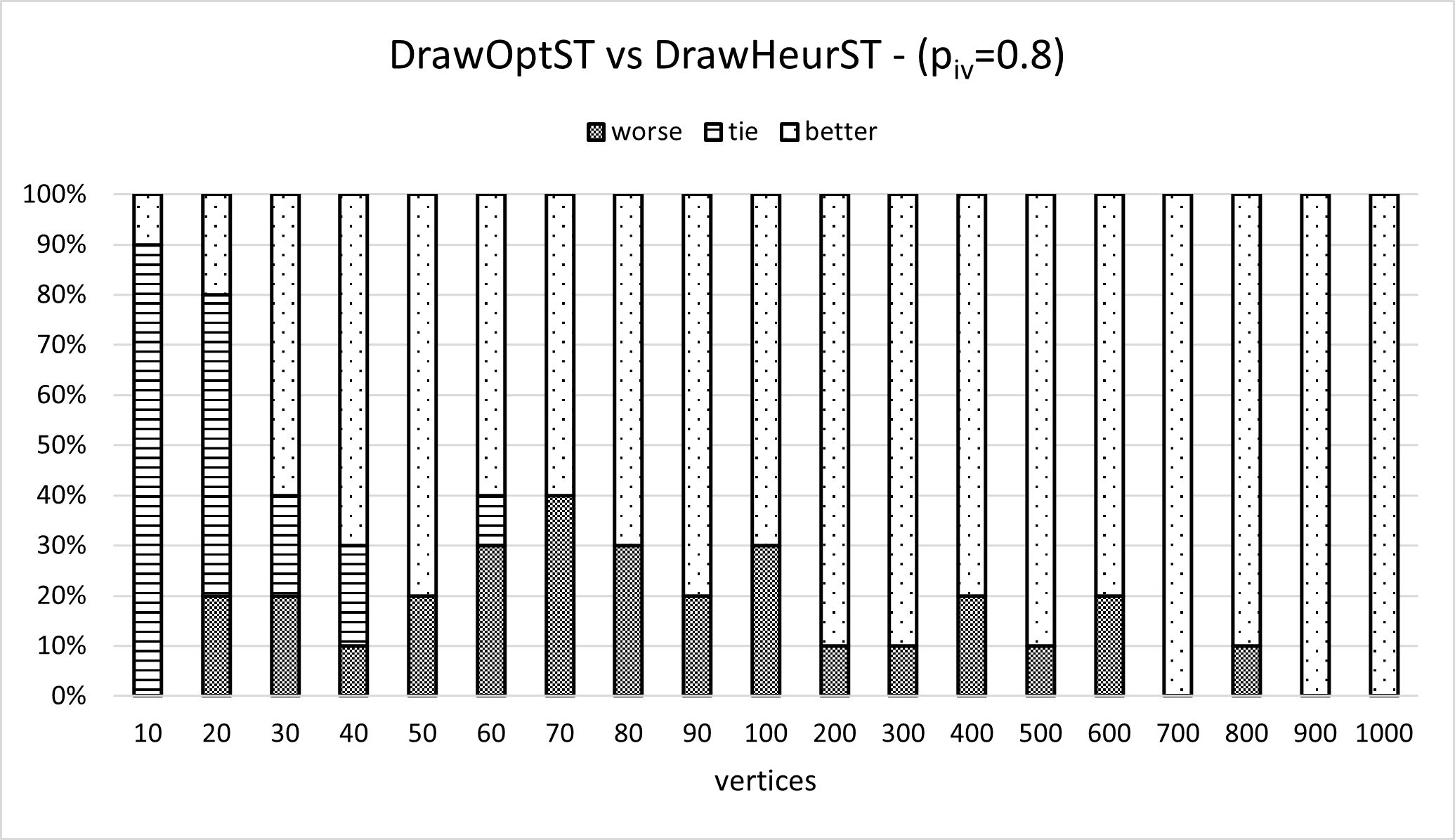}\label{fi:impInstances_0-8}}
	}\hfill
	\subfigure[]{
		{\includegraphics[width=0.475\textwidth]{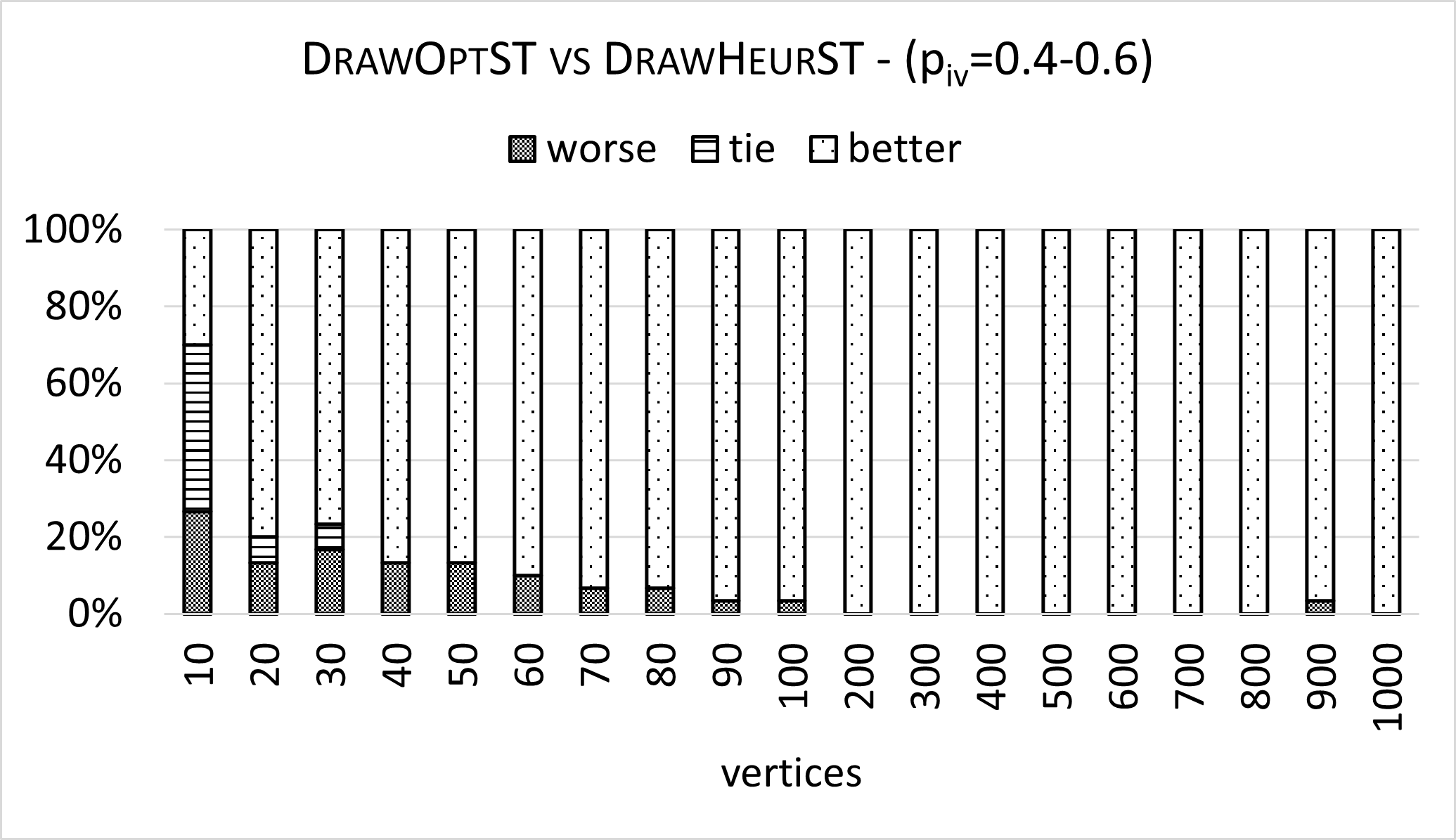}\label{fi:impInstances_0-6_0-4}}
	}\hfill
	\subfigure[]{
		{\includegraphics[width=0.475\textwidth]{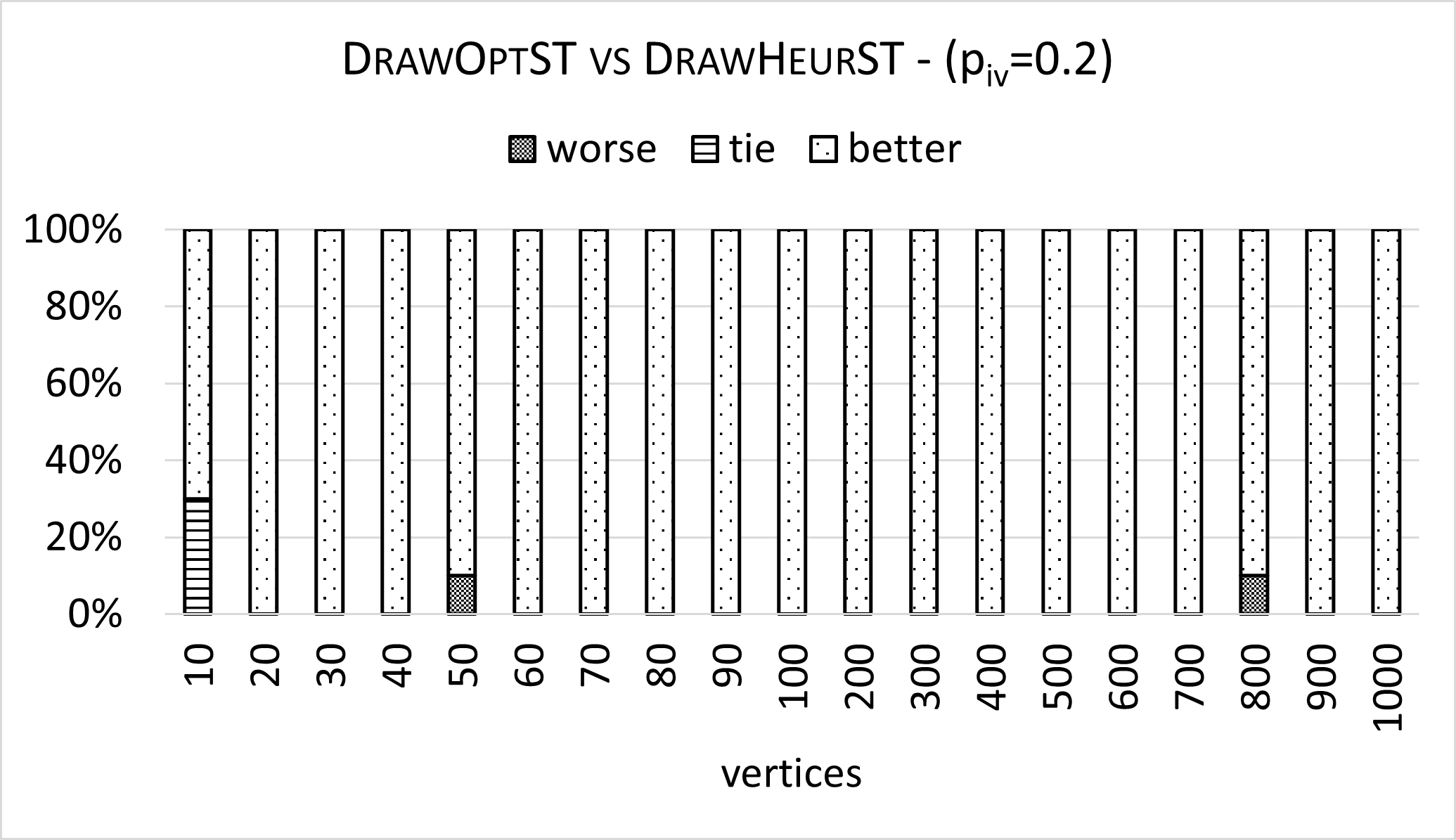}\label{fi:impInstances_0-2}}
	}\caption{Instances for which \textsc{DrawOptST} produces drawings that are more compact than \textsc{DrawHeurST} (label ``better'').}\label{fi:impInstances}
\end{figure}

% Area improvement
\begin{figure}[tb]
	\centering
	\subfigure[]{
		{\includegraphics[width=0.475\textwidth]{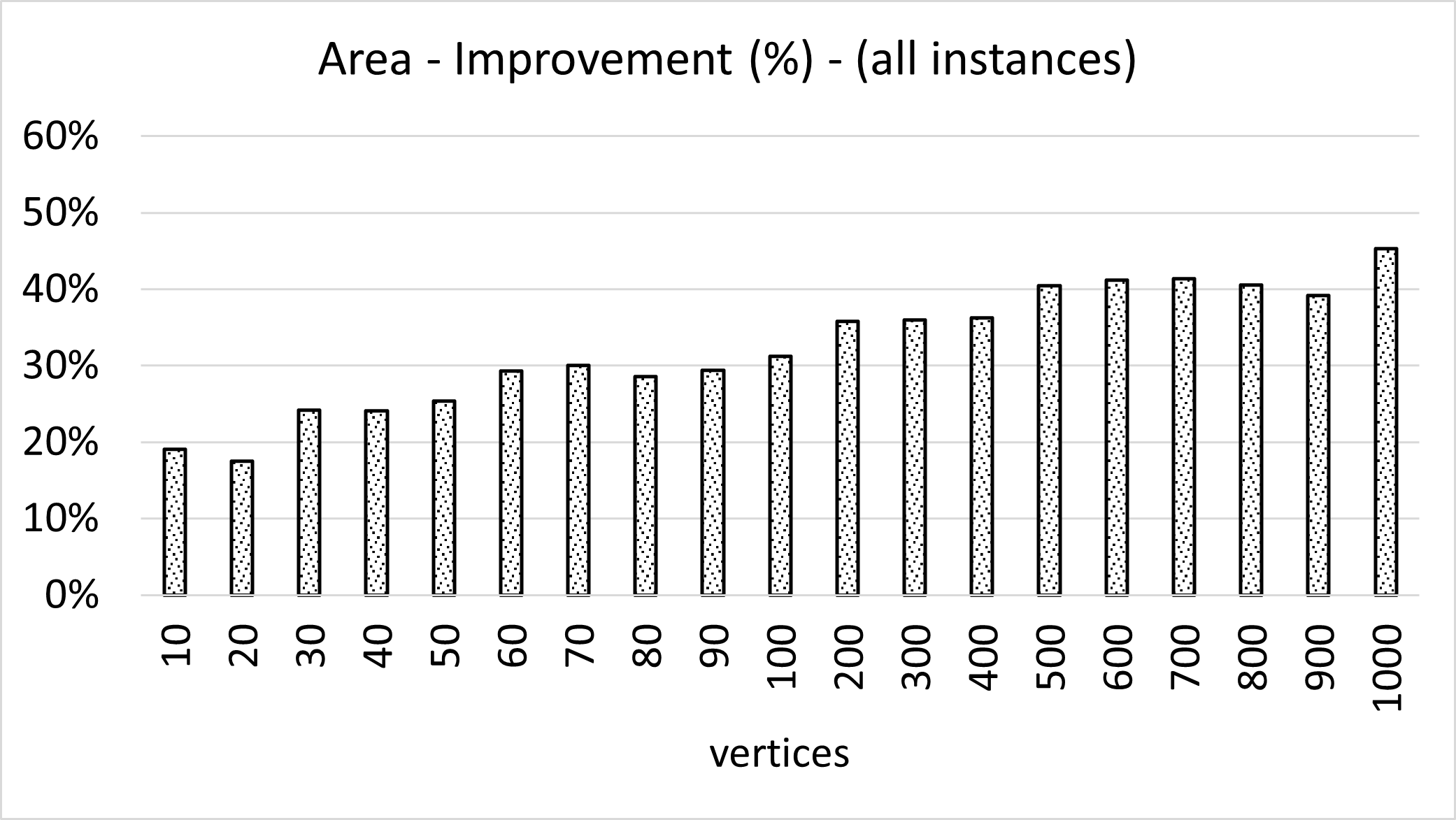}\label{fi:impArea_all}}
	}
	\subfigure[]{
		{\includegraphics[width=0.475\textwidth]{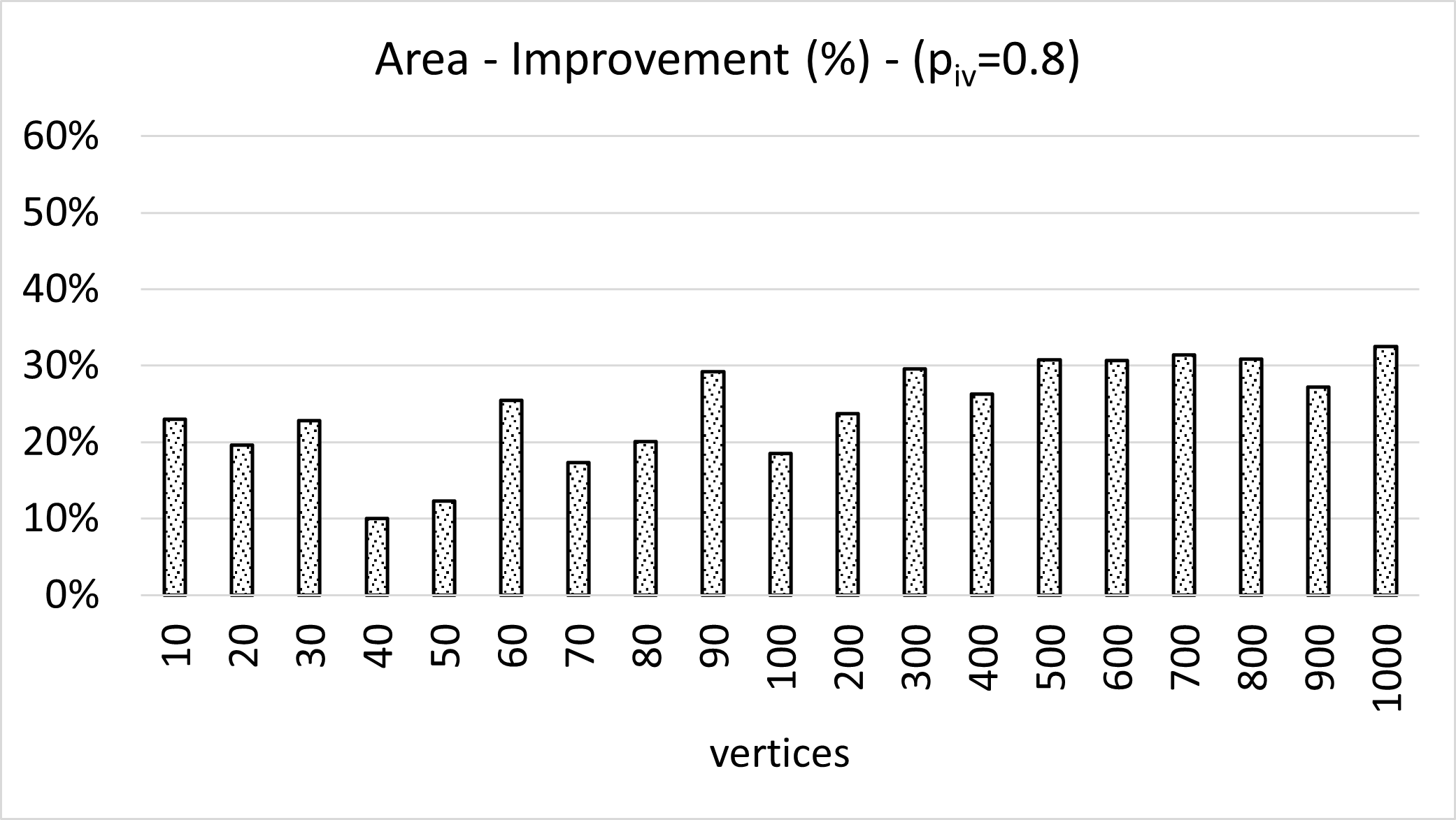}\label{fi:impArea_0-8}}
	}\hfill
	\subfigure[]{
		{\includegraphics[width=0.475\textwidth]{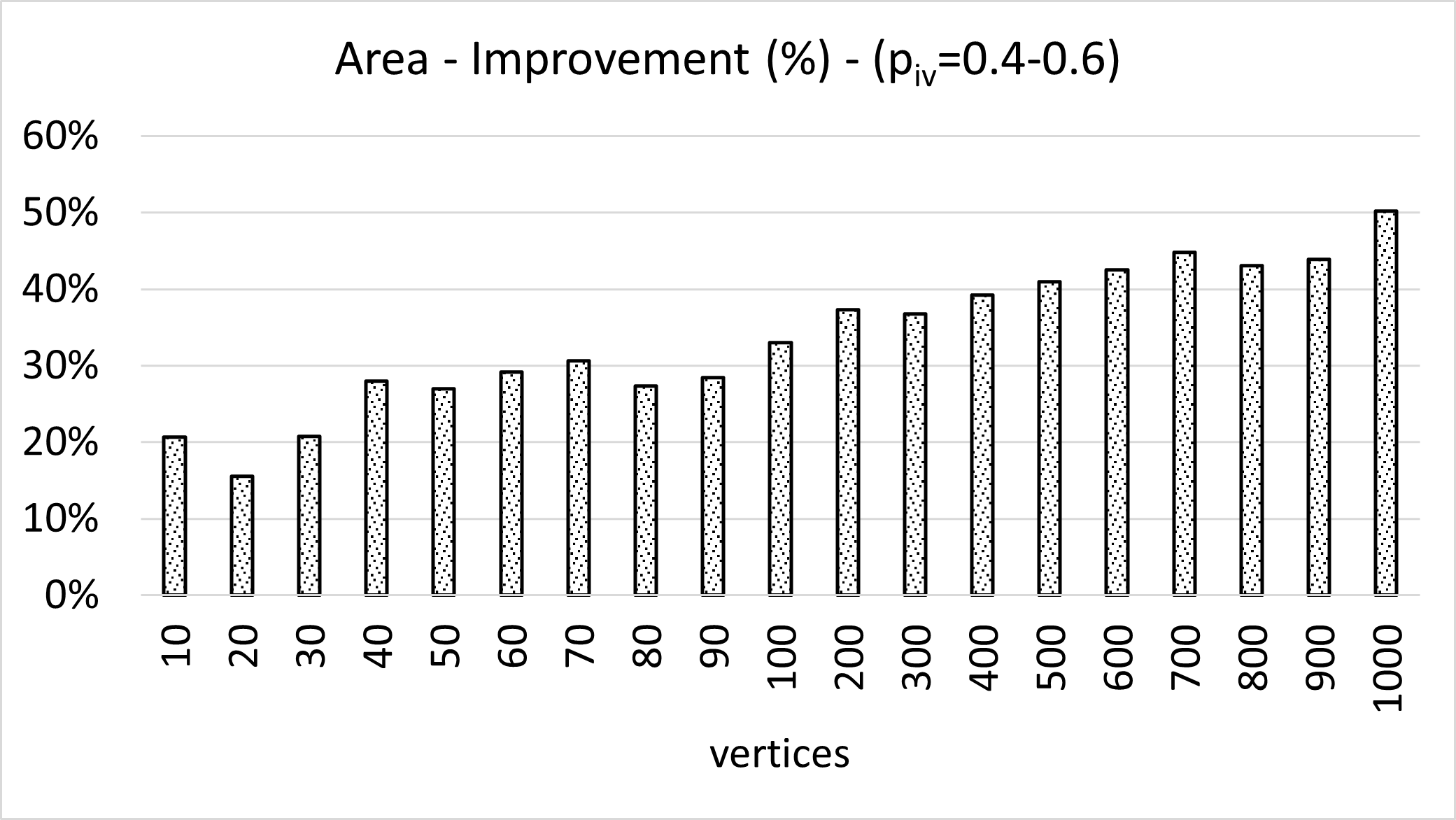}\label{fi:impArea_0-6_0-4}}
	}\hfill
	\subfigure[]{
		{\includegraphics[width=0.475\textwidth]{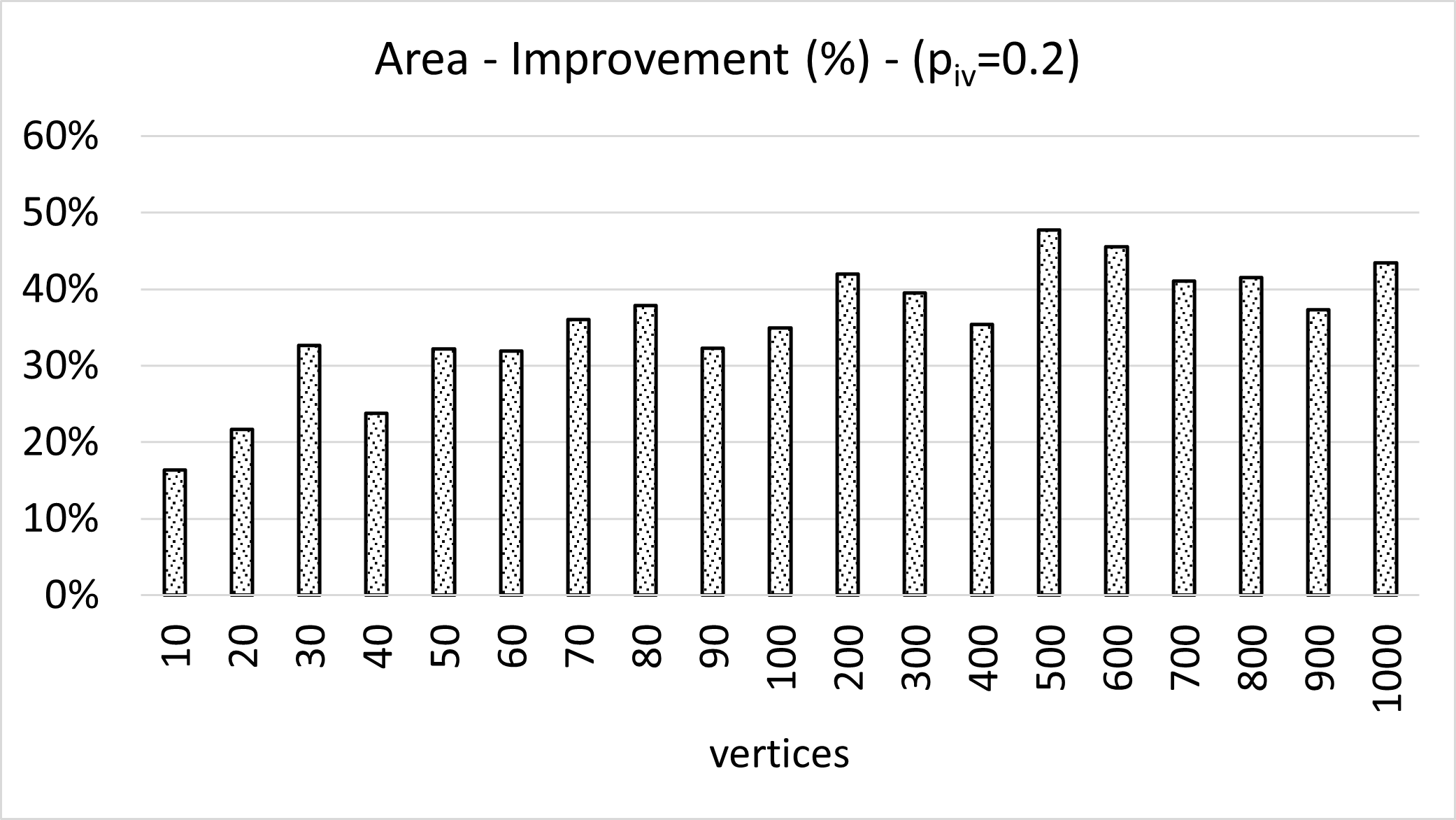}\label{fi:impArea_0-2}}
	}\caption{Area improvement (\%) of \textsc{DrawOptST} w.r.t. \textsc{DrawHeurST}, for the instances where \textsc{DrawOptST} is ``better'' (i.e., the ``better'' instances in \cref{fi:impInstances}).}\label{fi:impArea}
\end{figure}

% correlation trans-area
\begin{figure}[t]
	\centering
	\subfigure[]{
		{\includegraphics[width=0.475\textwidth]{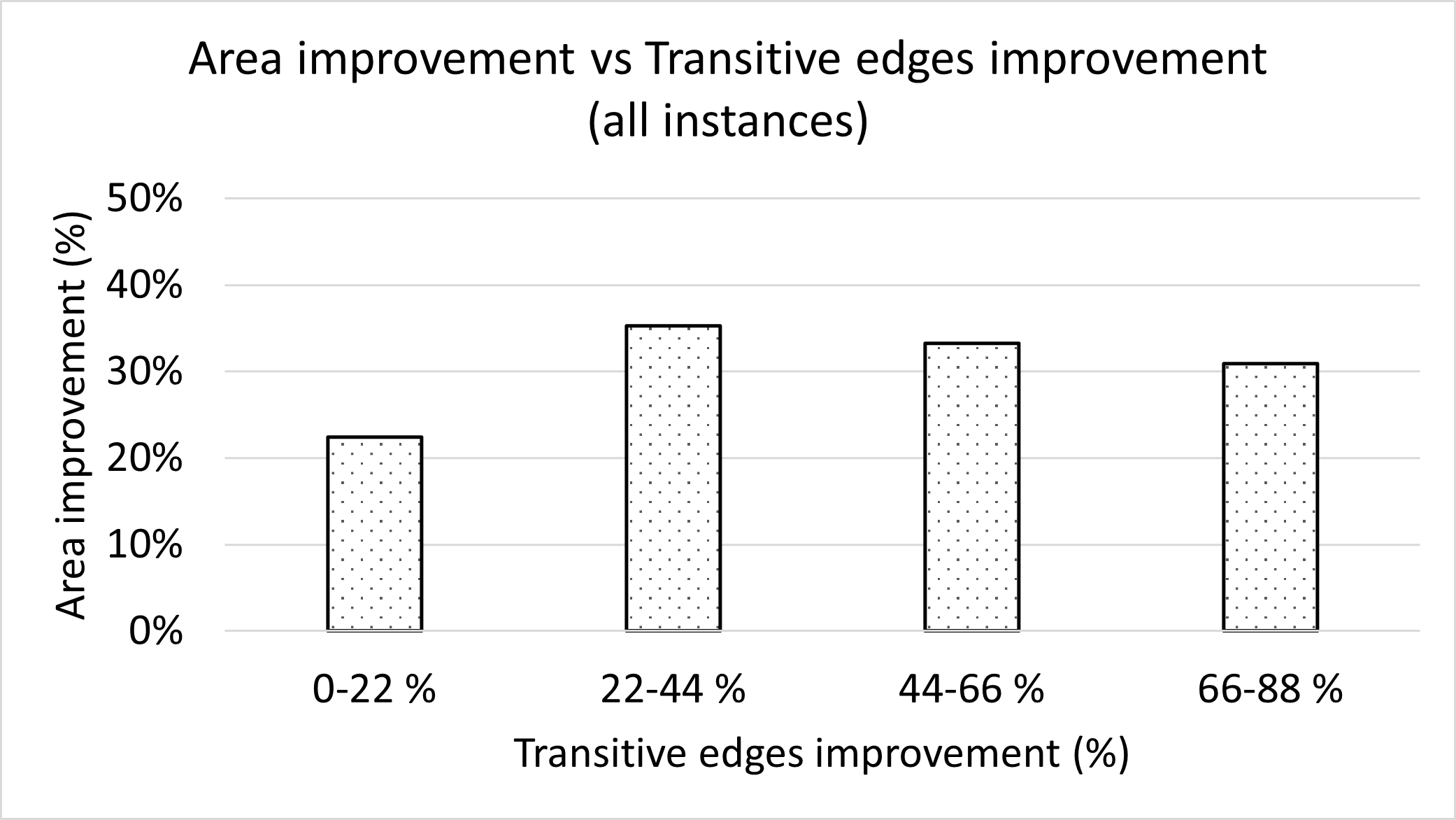}\label{fi:Corr_TrE_Area}}
	}
	\subfigure[]{
		{\includegraphics[width=0.475\textwidth]{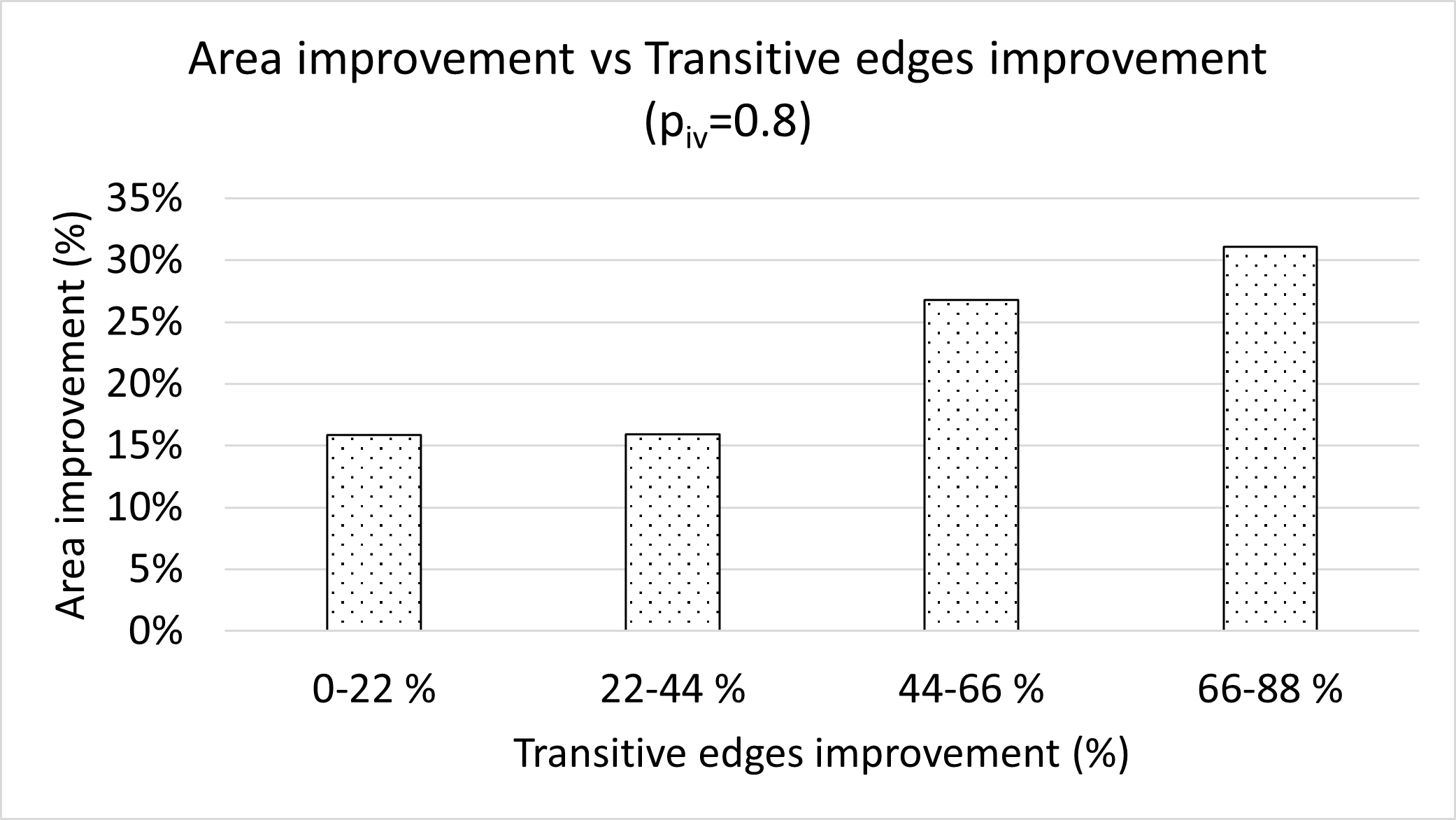}\label{fi:Corr_TrE_Area_0-8}}
	}\hfill
	\subfigure[]{
		{\includegraphics[width=0.475\textwidth]{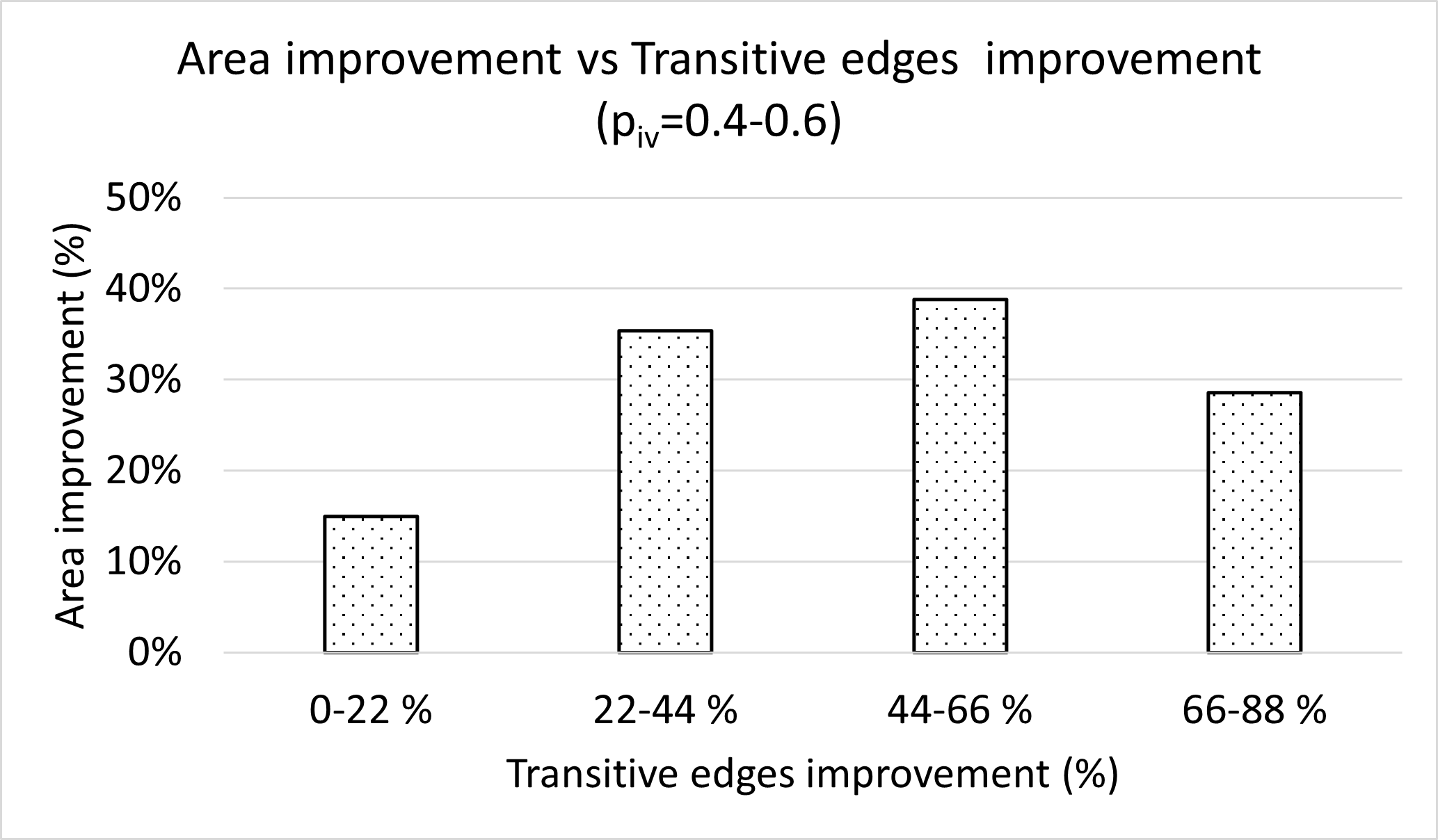}\label{fi:Corr_TrE_Area_0-6_0-4}}
	}\hfill
	\subfigure[]{
		{\includegraphics[width=0.475\textwidth]{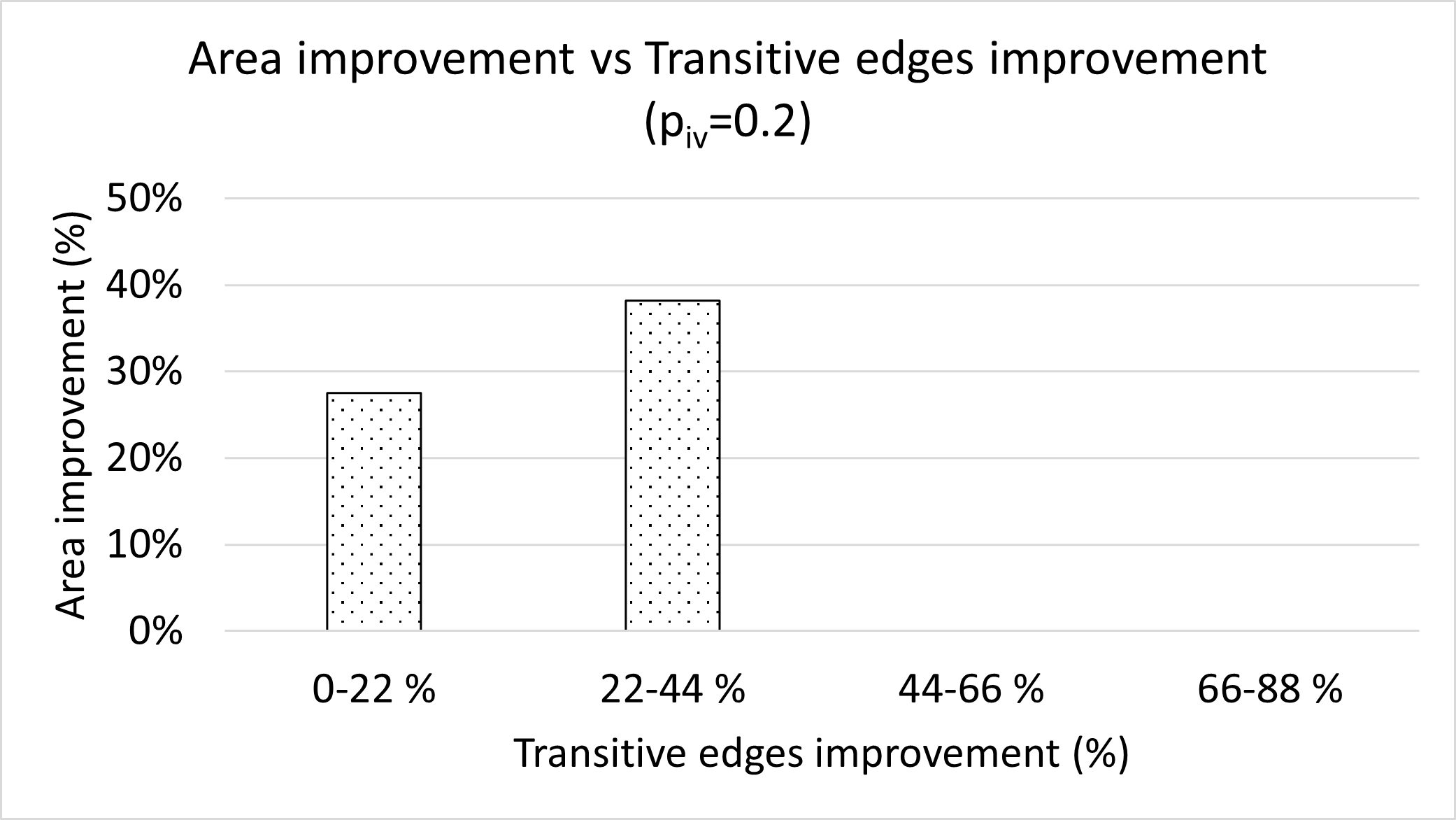}\label{fi:Corr_TrE_Area_0-2}}
	}\caption{Correlation between the improvement (reduction) in terms of drawing area and in terms of transitive edges improvement.}\label{fi:correlation}
\end{figure}

About (G2), \cref{fi:impTrE} shows the reduction (in percentage) of the number of transitive edges in the solutions of \textsc{OptST} with respect to the solutions of \textsc{HeurST}. More precisely, \cref{fi:impTrE_all} reports values averaged over all instances with the same number of vertices; \cref{fi:impTrE_0-8}, \cref{fi:impTrE_0-6_0-4}, and \cref{fi:impTrE_0-2} report the same data, partitioning the instances by different values of $p_{\rm iv}$, namely $0.8$ (the sparsest instances), $0.4$-$0.6$ (instances of medium density), and $0.2$ (the densest instances). 
For each instance, denoted by ${\rm trOpt}$ and ${\rm trHeur}$ the number of transitive edges of the solutions computed by \textsc{OptST} and \textsc{HeurST}, respectively, the reduction percentage equals the value $\Big(\frac{{\rm trHeur} - {\rm trOpt}}{\max\{1,{\rm trHeur}\}} \times 100 \Big)$. Over all instances, the average reduction is about $35\%$; it grows above $60\%$ on the larger graphs if we restrict to the sparsest instances (with improvements larger than $80\%$ on some graphs), while it is below $30\%$ for the densest instances, due to the presence of many 3-cycles, for which a transitive edge cannot be avoided.

About~(G3), \cref{fi:impInstances} shows the percentage of instances for which \textsc{DrawOptST} produces drawings that are better than those produced by \textsc{DrawHeurST} in terms of area requirement (the label ``better'' of the legend). It can be seen that \textsc{DrawOptST} computes more compact drawings for the majority of the instances. In particular, it is interesting to observe that this is most often the case even for the densest instances (i.e., those for $p_{\rm iv}=0.2$), for which we have previously seen that the average reduction of transitive edges is less evident. 
For those instances for which \textsc{DrawOptST} computes more compact drawings than \textsc{DrawHeurST}, \cref{fi:impArea} reports the average percentage of improvement in terms of area requirement (i.e., the percentage of area reduction). The values are mostly between $30\%$ and $50\%$. To complement this data, \cref{fi:correlation} reports the trend of the improvement (reduction) in terms of drawing area with respect to the reduction of the transitive edges (discretized in four intervals). For the instances with $p_{\rm iv}=0.8$ and $p_{\rm iv}=0.2$, the correlation between these two measures is quite evident. For the instances of medium density ($p_{\rm iv} \in \{0.4, 0.5, 0.6\}$), the highest values of improvement in terms of area requirement are observed for reductions of transitive edges between $22 \%$ and $66 \%$. Figures~\ref{fi:ug_8_100_08_polyline} and ~\ref{fi:ug_6_100_05_polyline} in the appendix show drawings computed by \textsc{DrawHeurST} and \textsc{DrawOptST} for two of our~instances. 
%Overall, the experimental data provide clear indications and good results for each of the three experimental goals~(G1),~(G2), and~(G3). 
%\textcolor{red}{The data are publicly available at \url{http://xxxxx}}. 

%\medskip
%GRAFICI
%\begin{itemize}
%	\item G1 - running time così come sono (pox plots) \cref{fi:box-plot}
%	\item G2 - improvement transitive edges (all, [0.8], [0.6, 0.5, 0.4], [0.2]) \cref{fi:impTrE}
%	\item G3 - improvement area (all, [0.8], [0.6, 0.5, 0.4], [0.2]) \cref{fi:impArea}; percentage of graphs on which we actually improve (all, [0.8], [0.6, 0.5, 0.4], [0.2]) \cref{fi:impInstances}; correlation (perct. area improv. - perc. transitive edges improv. with 3 or 4 classi) \cref{fi:correlation}. 
%\end{itemize}

%\begin{figure}[ht]
%	\centering		{\includegraphics[width=0.70\textwidth]{figures/Corr_TrE_Area}\label{fi:Corr_TrE_Area}}
%	\caption{(a) Correlation between area improvement and transitive edges improvement}\label{fi:correlation}
%\end{figure}

\section{Final Remarks and Open Problems}\label{se:conclusions}

We addressed the problem of computing $st$-orientations with the minimum number of transitive edges. This problem has practical applications in graph drawing, as finding an $st$-orientation is at the heart of several graph drawing algorithms. Although $st$-orientations without transitive edges have been studied from a combinatorial perspective~\cite{DBLP:journals/corr/abs-2105-06955}, there is a lack of practical algorithms, and the complexity of deciding whether a graph can be oriented to become an $st$-graph without transitive edges seems not to have been previously addressed. 

We proved that this problem is NP-hard in general and we described an ILP model for planar graphs based on characterizing planar $st$-graphs without transitive edges in terms of a constrained labeling of the vertex angles inside its faces.
An extensive experimental analysis on a large set of instances shows that our model is fast in practice, taking few seconds for graphs of thousand vertices. It saves on average $35\%$ of transitive edges w.r.t. a classical algorithm that computes an unconstrained $st$-orientation. We also showed that for classical layout algorithms that compute polyline drawings of planar graphs through an $st$-orientation, minimizing the number of transitive edges yields more compact drawings most of the time (see also \cref{fi:ug_8_100_08_polyline} and \cref{fi:ug_6_100_05_polyline} in the appendix).

\noindent We suggest two future research directions:
%\begin{itemize}
$(i)$ It remains open to establish the time complexity of the problem for planar graphs. Are there polynomial-time algorithms that compute $st$-orientations with the minimum number of transitive edges for all planar graphs or for specific subfamilies of planar graphs?
%
%To this aim, we remark that relaxing the integrity constraint of our ILP model often leads to invalid solutions; this implies that the constraint matrix associated with the model is not unimodular 
%
$(ii)$ One can extend the experimental analysis to real-world graphs and design fast heuristics, which can be compared to the optimal algorithm. 

%\textsf{(P2)} It is interesting to study the problem for specific subfamilies of planar graphs, such as for example outerplanar graphs, series-parallel graphs, or graphs of bounded vertex-degree.
%\end{itemize}

%
%
%Maybe in this section we can also make some observations about planar graphs. For example, we can show examples of planar graphs with no triangular faces that require some transitive edges. We can mention efficient algorithms for specific classes of planar graphs: quadrangulations (this is trivial); internal quadrangulations (less trivial, but still doable in linear time).  
%Mention some open problems.

\clearpage

\bibliographystyle{splncs04}
\bibliography{bibliography}

\newpage
\appendix

\section{Appendix}\label{se:app}

\subsection{Additional Material for \Cref{se:hardness}}\label{sse:app-hardness}

\begin{restatable}{observation}{obInducedPath}\label{ob:induced-path}
Let $(v_1,v_2,\dots,v_k)$ be a path of $G$ such that its internal vertices $v_2, v_3, \dots, v_{k-1}$ have degree $2$ in $G$ and are different from $s$ and $t$. In any non-transitive $st$-orientation of $G$ the edges $(v_i,v_{i+1})$, with $i=1, \dots, k-1$, are all directed from $v_i$ to $v_{i+1}$ or they are all directed from $v_{i+1}$ to $v_i$.     
\end{restatable}
%\obInducedPath*
\begin{proof}
The statement can be easily proved by observing that if two edges of the path have an inconsistent orientation (as in \cref{fi:observations-c}) then the path would contain an internal vertex that is a source or a sink different from $s$ and $t$, contradicting the hypothesis that the orientation is an $st$-orientation.
\end{proof}

\begin{figure}[h!]
	\centering
	\subfigure[]{
		\includegraphics[width=0.21\textwidth, page=1]{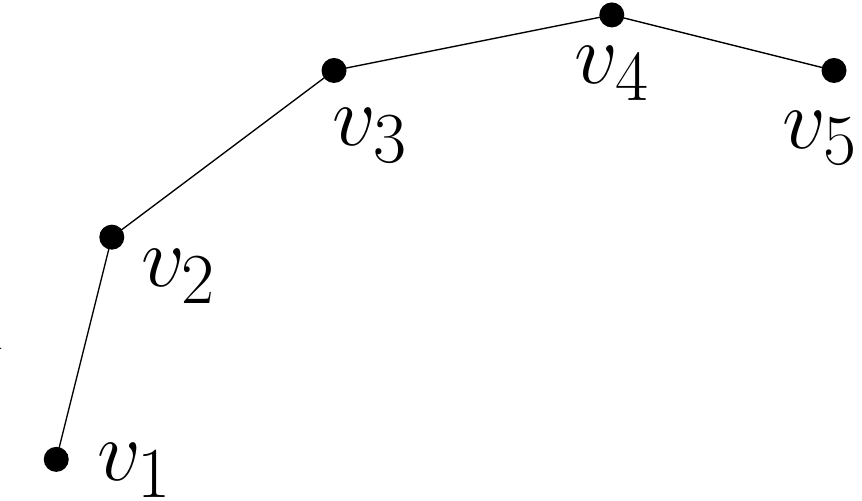}
		\label{fi:observations-a}
	}
	\hfil
	\subfigure[]{
		\includegraphics[width=0.21\textwidth, page=2]{observations}
		\label{fi:observations-b}
	}
	\hfil
	\subfigure[]{
		\includegraphics[width=0.21\textwidth, page=3]{observations}
		\label{fi:observations-c}
	}
	\hfil
	\subfigure[]{
		\includegraphics[width=0.21\textwidth, page=4]{observations}
		\label{fi:observations-d}
	}
	\caption{(a) A path of $G$ with all internal vertices of degree two. (b) A consistent orientation of the path. (c) An inconsistent orientation of the path generates sinks or sources. (d) A directed path of $G$ and a chord.}\label{fi:observations} 
\end{figure}

\begin{restatable}{observation}{obChord}\label{ob:chord}
Let $(v_1,v_2,\dots,v_k)$ be a path of $G$ and let $(v_1,v_k)$ be an edge of $G$. In any non-transitive $st$-orientation of $G$ the edges $(v_i,v_{i+1})$, with $i=1, \dots, k-1$, cannot be all directed from $v_i$ to $v_{i+1}$.     
\end{restatable}
%\obChord*
\begin{proof}
Suppose for a contradiction that there exists a non-transitive $st$-orientation of $G$ such that each edge $(v_i,v_{i+1})$, with $i=1, \dots, k-1$, is directed from $v_i$ to $v_{i+1}$ (refer to \cref{fi:observations-d}). If edge $(v_1,v_k)$ was also directed from $v_1$ to $v_k$ it would be a transitive edge, contradicting the hypothesis that the orientation is non-transitive. Otherwise, if $(v_1,v_k)$ was directed from $v_k$ to $v_1$ it would form a directed cycle, contradicting the hypothesis that the orientation is an $st$-orientation.      
\end{proof}

\subsubsection{Proof of \Cref{le:fork-gadget}}

\leForkGadget*
\begin{proof}
Suppose edge $e_1$ is oriented entering $F$ (refer to \cref{fi:gadget-b}). One between $e_9$ or $e_{10}$ must be oriented exiting $F$, otherwise $F$ contains a sink contradicting the fact that we have an $st$-orientation of $G$. Since gadget $F$ is symmetric, we may assume without loss of generality that edge $e_9$ is oriented exiting $F$. Therefore, there must be at least one directed path from $e_1$ to $e_9$ traversing $F$. There are three possible such directed paths: (1) path $(e_1,e_4,e_8,e_7,e_6,e_9)$; (2) path $(e_1,e_3,e_6,e_9)$; and (3) path $(e_1,e_2,e_5,e_9)$. 
Suppose Case (1) applies, i.e., $(e_1,e_4,e_8,e_7,e_6,e_9)$ is a directed path. We have a contradiction because of \cref{ob:chord} applied to the directed path $(e_4,e_8,e_7)$ and the chord $e_3$. Suppose Case (2) applies, i.e., $(e_1,e_3,e_6,e_9)$ is a directed path. Note that by \cref{ob:induced-path} the edges $e_2$ and $e_5$ must be both directed in the same direction. If they were directed towards $v$, then we would have a directed cycle $(e_3,e_6,e_5,e_2)$. Hence, $(e_2,e_5)$ are directed away from $v$ and, since $(e_1,e_2,e_5,e_9)$ is also a directed path, Case (2) implies Case (3). Conversely, suppose Case (3) applies, i.e., $(e_1,e_2,e_5,e_9)$ is a directed path. Edge $e_6$ must be directed towards $w$. In fact, if $e_6$ was directed away from $w$ we would have a contradicton by \cref{ob:chord} applied to the directed path $(e_2,e_5,e_6)$ and the chord $e_3$. Also, edge $e_3$ must be directed away from $v$. In fact, if $e_3$ was directed towards $v$ edge $e_6$ would be a transitive edge with respect to the directed path $(e_3,e_2,e_5)$. It follows that $(e_1,e_3,e_6,e_9)$ would also be a directed path and Case (3) implies Case (2). Therefore, we have to assume that Case (2) and Case (3) both apply. Note that by \cref{ob:induced-path} the edges $e_4$ and $e_8$ must be both directed in the same direction. If the path $(e_8,e_4)$ was oriented exiting $z$ and entering $v$ then we would have a contradiction because of \cref{ob:chord} applied to the directed path $(e_8,e_4,e_3)$ and the chord $e_7$. It follows that the path $(e_4,e_8)$ is oriented exiting $v$ and entering $z$. Now, edge $e_7$ must be oriented entering $z$, otherwise $e_3$ would be a transitive edge with respect to the path $(e_4,e_8,e_7)$. Finally, edge $e_{10}$ must be oriented exiting $z$, otherwise $z$ would be a sink. 
In conclusion, if $e_1$ is oriented entering $F$, then $e_9$ and $e_{10}$ must be oriented exiting $F$. 

With analogous and symmetric arguments it can be proved that if $e_1$ is oriented exiting $F$ (refer to \cref{fi:gadget-c}), then $e_9$ and $e_{10}$ must be oriented entering $F$. Since $e_1$ must be oriented in one way or the other, the only two possible orientations of $F$ are those depicted in \cref{fi:gadget-b,fi:gadget-c} and the statement follows. 
\end{proof}

\subsubsection{Proof of \Cref{le:variable-gadget}}

%\begin{figure}[tb]
%	\centering
%	\subfigure[]{
%		\includegraphics[width=0.40\textwidth, page=1]{variable-gadget}
%		\label{fi:variable-gadget-a}
%	}
%	\hfil
%	\subfigure[]{
%		\includegraphics[width=0.40\textwidth, page=2]{variable-gadget}
%		\label{fi:variable-gadget-b}
%	}
%	\caption{The variable gadget $V_x$ and its \texttt{true} (a) and \texttt{false} (b) orientations.}\label{fi:variable-gadget} 
%\end{figure}

\leVariableGadget*
\begin{proof}
Suppose edge $e_1$ of $F_x$ is oriented entering $F_x$ (see \cref{fi:variable-gadget-a}). By \cref{le:fork-gadget} edge $x$ is oriented exiting $F_x$ and, hence, exiting $V_x$. Also edge $e_9$ of $F_x$, which coincides with $e_{10}$ of $F_{\overline{x}}$, is oriented exiting $F_x$ and entering $F_{\overline{x}}$. Now, always by \cref{le:fork-gadget}, edge $e_1$ of $F_{\overline{x}}$ is oriented exiting $F_{\overline{x}}$ and edge $e_9$ of $F_{\overline{x}}$, which coincides with edge $\overline{x}$ of $V_x$, is oriented entering $F_{\overline{x}}$ and, hence, entering $V_x$. 

Suppose now that edge $e_1$ of $F_x$ is oriented exiting $F_x$ (see \cref{fi:variable-gadget-b}). By \cref{le:fork-gadget} edge $x$ is oriented entering $F_x$ and, hence, entering $V_x$. Also edge $e_9$ of $F_x$, which coincides with $e_{10}$ of $F_{\overline{x}}$, is oriented entering $F_x$ and exiting $F_{\overline{x}}$. Now, always by \cref{le:fork-gadget}, edge $e_1$ of $F_{\overline{x}}$ is oriented entering $F_{\overline{x}}$ and edge $e_9$ of $F_{\overline{x}}$, which coincides with edge $\overline{x}$ of $V_x$, is oriented exiting $F_{\overline{x}}$ and, hence, exiting $V_x$. 
Finally, observe that, even if a directed path was added outside $V_x$ from edge $x$ to edge $\overline{x}$ or vice versa, no directed cycle traverses $V_x$. In fact, all directed paths exiting $V_x$ originate from $s$ and all directed paths entering $V_x$ go to $t$.   
\end{proof}

\subsubsection{Proof of \Cref{th:hardness}}

\thHardness*
\begin{proof}
The reduction from an instance $\varphi$ of \textsc{NAE3SAT} to an instance~$I_\varphi$ previously described is performed in time linear in the size of $\varphi$.

Suppose $I_\varphi = \langle G, s, t \rangle$ is a positive instance of \textsc{NTO} and consider any non-transitive $st$-orientation of $G_\varphi$. Consider a clause $c$ of $\varphi$ and the corresponding vertex $v_c$ in $G$. Since vertex $v_c$ is not a sink nor a source it must have at least one entering edge $e_\textrm{in}$ and at least one exiting edge $e_\textrm{out}$. Consider first edge $e_\textrm{in}$ and assume it corresponds to a directed literal $x_i$ of $c$ (to a negated literal $\overline{x}_i$ of $c$, respectively). By construction, edge $e_\textrm{in}$ comes from the edge $x_i$ (edge $\overline{x}_i$, respectively) of variable gadget $V_{x_i}$ or from an intermediate split gadget $S_{x_i}$ ($S_{\overline{x}_i}$, respectively) that has edge $x_i$ (edge $\overline{x}_i$, respectively) as input edge. Therefore, by \cref{le:variable-gadget,le:split-gadget} edge $x$ (edge $\overline{x}_i$, respectively) of $V_{x_i}$ is oriented exiting $V_{x_i}$, which corresponds to a \texttt{true} literal of $c$.
Now consider edge $e_\textrm{out}$ and assume it corresponds to a directed literal $x_j$ of $c$ (to a negated literal $\overline{x}_j$ of $c$, respectively). With analogous arguments as above you conclude that edge $x_j$ (edge $\overline{x}_j$, respectively) of $V_{x_j}$ is oriented entering $V_{x_j}$, which corresponds to a \texttt{false} literal of $c$. Therefore, each clause $c$ has both a \texttt{true} and a \texttt{false} literal and the \textsc{NAE3SAT} instance $\varphi$ is a yes instance. 

Conversely, suppose that instance $\varphi$ is a yes instance of \textsc{NAE3SAT}. Consider a truth assignment to the variables in $X$ that satisfies $\varphi$. Orient the edges of each variable gadget $V_x$ as depicted in \cref{fi:variable-gadget-a} or \cref{fi:variable-gadget-b} depending on whether variable $x$ is set to \texttt{true} or \texttt{false} in the truth assignment, respectively.
Orient each split gadget according to its input edge. Since the truth assignment is such that every clause has a \texttt{true} literal and a \texttt{false} literal, the corresponding clause gadget $C_c$  will have at least one incoming edge and one outgoing edge. Therefore the obtained orientation is a non-transitive $st$-orientation of~$G$.
Regarding acyclicity, observe that variable gadgets and clause gadgets whose edges are oriented as depicted in \cref{fi:variable-gadget} and \cref{fi:clause-gadget}, respectively, are acyclic. Also, a split gadget whose output edges are oriented all exiting or all entering the gadget is acyclic. Since all the directed paths that enter a variable gadget $V_{x_i}$ terminate at $t$ without exiting $V_{x_i}$ and all the directed paths that leave $V_{x_i}$ come from $s$ without entering $V_{x_i}$, there cannot be a directed cycle involving a variable gadget $V_{x_i}$. It remains to show that there are no directed cycles involving split gadgets and clause gadgets. However, by \cref{le:split-gadget} no directed path may enter a split gadget from a clause gadget and exit the split gadget towards a second clause gadget. Hence, directed cycles involving clause gadgets and split gadgets alone cannot exist.  

Finally, \textsc{NTO} is trivially in NP, as one can non-deterministically explore all possible orientations of the~graph.  
\end{proof}

\subsubsection{Complexity of \textsc{NTO} where $s$ and $t$ can be freely chosen.}

Observe that the variant of the \textsc{NTO} problem where the source and the target vertices of $G$ are not prescribed but can be freely choosen is also NP-hard. Problem \textsc{NTO}, in fact, can be easily reduced to it. Consider an instance $\langle G^*, s^*, t^* \rangle$ of \textsc{NTO}. Add two vertices $s^+$ and $t^+$ to $G^*$ and connect them to $s^*$ and to $t^*$, respectively. Call $G^+$ the obtained graph. Since $s^+$ and $t^+$ have degree one in $G^+$, in any non-transitive $st$-orientation of $G^+$ they can only be sources or sinks, where if one of them is the source the other one is the sink. Hence, given any non-transitive $st$-orientation of $G^+$ you can immediately find a non-transitive $s^*t^*$-orientation of $G^*$, possibly by reversing all edge orientations if $t^+$ is the source and $s^+$ is the sink. Conversely, given a non-transitive $s^*t^*$-orientation of $G^*$ you easily find an $st$-orientation of $G$ orienting the edge $(s^+,s^*)$ from $s^+$ to $s^*$ and the edge $(t^*,t^+)$ from $t^*$ to $t^+$.
Therefore, the addition of edges $(s^+,s^*)$ and $(t^+,t^*)$ is a polynomial-time reduction from problem \textsc{NTO} with prescribed source and target to the variant of the \textsc{NTO} problem where these vertices can be freely choosen, proving the hardness of the latter problem. Since this variant of \textsc{NTO} is also trivially in NP it is NP-complete.

\subsection{Additional Material for \Cref{se:experiments}}

\begin{figure}[htb]
	\centering
	\subfigure[14 transitive edges]{
		\includegraphics[width=0.4\textwidth]{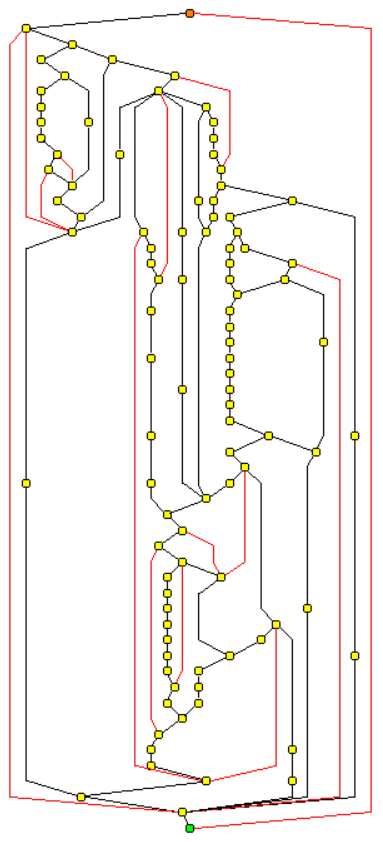}
		\label{fi:ug_8_100_08_heur}
	}
	\hfil
	\subfigure[7 transitive edges]{
		\includegraphics[width=0.4\textwidth]{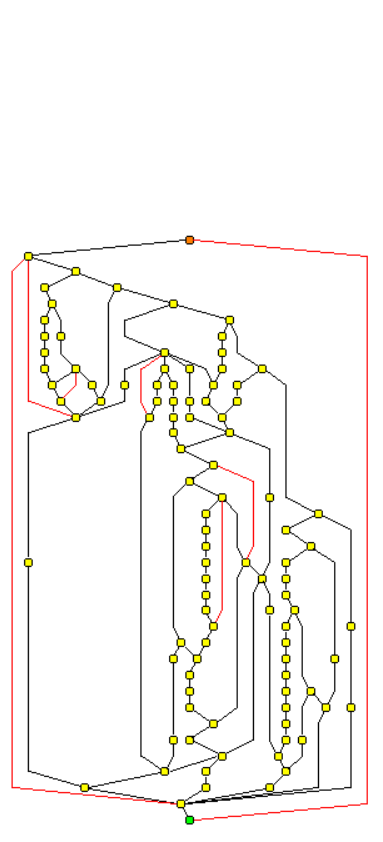}
		\label{fi:ug_8_100_08_opt}
	}
	\caption{Two polyline drawings of the same plane graph with $100$ vertices and $\rm p_{iv}=0.8$ computed by (a) \textsc{DrawHeurST} and (b) \textsc{DrawOptST}. Transitive edges are colored red. }\label{fi:ug_8_100_08_polyline}
\end{figure}

\begin{figure}[htb]
	\centering
	\subfigure[52 transitive edges]{
		\includegraphics[width=0.8\textwidth]{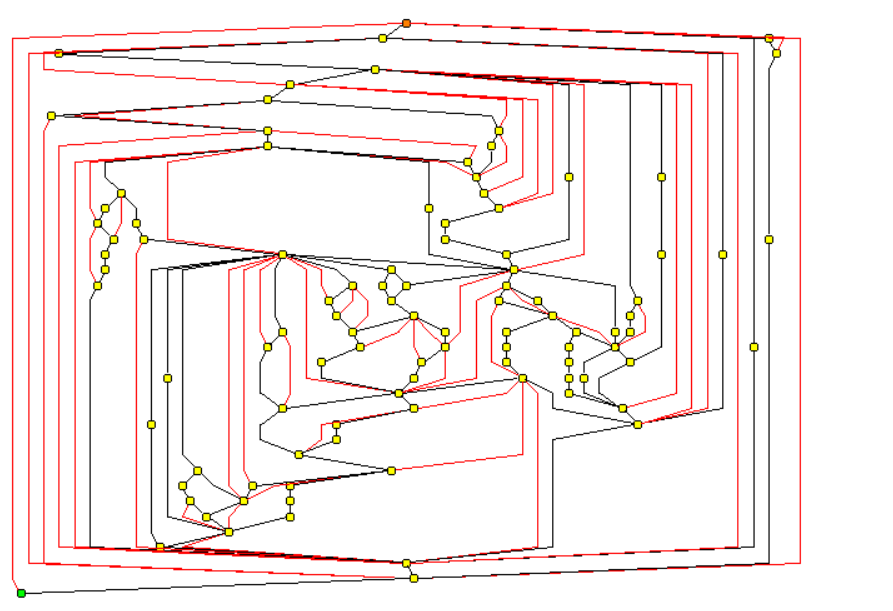}
		\label{fi:ug_6_100_05_heur}
	}
	\hfil
	\subfigure[37 transitive edges]{
		\includegraphics[width=0.8\textwidth]{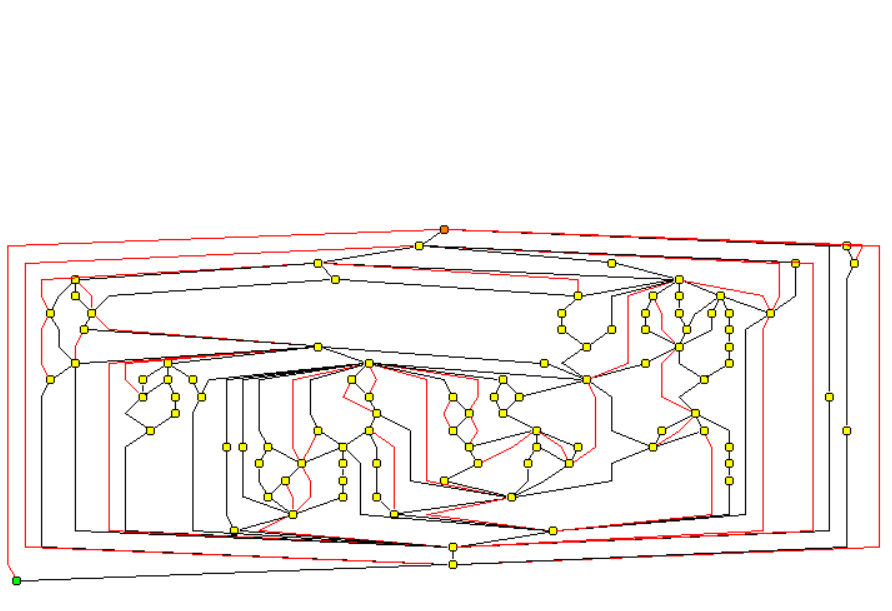}
		\label{fi:ug_6_100_05_opt}
	}
	\caption{Two polyline drawings of the same plane graph with $100$ vertices and $\rm p_{iv}=0.5$ computed by (a) \textsc{DrawHeurST} and (b) \textsc{DrawOptST}. Transitive edges are colored red. }\label{fi:ug_6_100_05_polyline}
\end{figure}

\begin{table}[t]
	\footnotesize
	\begin{center}
		\rotatebox{90}{
			\begin{tabular}{|c|c|c|c|c||c|c|c|c||c|c|c|c||c|c|c|c||c|c|c|c|}
				\cline{2-21}
				\multicolumn{1}{c|}{} &
				\multicolumn{4}{|c||}{0.8} & \multicolumn{4}{|c||}{0.6} & \multicolumn{4}{|c||}{0.5} & \multicolumn{4}{|c||}{0.4} & \multicolumn{4}{|c|}{0.2}\\    
				%\cline{1-21}
				\hline
				$n$ & AVG & MIN &  MAX & SD & AVG & MIN &  MAX & SD & AVG & MIN &  MAX & SD & AVG & MIN &  MAX & SD & AVG & MIN &  MAX & SD \\
				
				\hline
				{10} & 1.16 & 1.00 &  1.40 & 0.11 & 1.33 & 1.10 &  1.50 & 0.11 & 1.50 & 1.20 &  1.80 & 0.22 & 1.71 & 1.50 &  2.00 & 0.14 & 1.89 & 1.40 &  2.20 & 0.26   \\
				\hline
				{20} & 1.19 & 1.05 &  1.30 & 0.08 & 1.54 & 1.30 &  2.15 & 0.25 & 1.65 & 1.35 &  2.05 & 0.20 & 1.76 & 1.60 &  2.05 & 0.15 & 2.41 & 2.25 &  2.55 & 0.11   \\
				\hline
				{30} & 1.23 & 1.07 & 1.37 & 0.10 & 1.49 & 1.37 &  1.67 & 0.10 & 1.68 & 1.43 &  1.93 & 0.16 & 1.93 & 1.83 &  2.07 & 0.08 & 2.42 & 2.23 &  2.57 & 0.11   \\
				\hline
				{40} & 1.22 & 1.10 &  1.30 & 0.06 & 1.58 & 1.43 &  1.78 & 0.11 & 1.83 & 1.58 &  2.08 & 0.14 & 1.97 & 1.70 &  2.23 & 0.20 & 2.49 & 2.43 &  2.58 & 0.05   \\
				\hline
				{50} & 1.22 & 1.16 &  1.28 & 0.04 & 1.57 & 1.46 &  1.66 & 0.06 & 1.74 & 1.54 &  1.86 & 0.09 & 2.02 & 1.80 &  2.30 & 0.14 & 2.54 & 2.40 &  2.68 & 0.09   \\
				\hline
				{60} & 1.24 & 1.15 &  1.33 & 0.06 & 1.51 & 1.38 &  1.63 & 0.09 & 1.77 & 1.55 &  1.95 & 0.13 & 2.00 & 1.83 &  2.25 & 0.13 & 2.54 & 2.43 &  2.67 & 0.07   \\
				\hline
				{70} & 1.22 & 1.16 &  1.36 & 0.06 & 1.57 & 1.41 &  1.71 & 0.10 & 1.84 & 1.66 &  1.93 & 0.08 & 2.04 & 1.89 &  2.20 & 0.11 & 2.55 & 2.41 &  2.70 & 0.09   \\
				\hline
				{80} & 1.25 & 1.19 &  1.33 & 0.05 & 1.57 & 1.49 &  1.68 & 0.06 & 1.71 & 1.63 &  1.79 & 0.05 & 2.03 & 1.79 &  2.18 & 0.14 & 2.54 & 2.44 &  2.65 & 0.07   \\
				\hline
				{90} & 1.24 & 1.16 &  1.33 & 0.06 & 1.54 & 1.40 &  1.71 & 0.10 & 1.80 & 1.67 &  1.96 & 0.11 & 2.05 & 1.93 &  2.17 & 0.08 & 2.59 & 2.42 &  2.76 & 0.10   \\
				\hline
				{100} & 1.25 & 1.15 &  1.34 & 0.05 & 1.53 & 1.40 &  1.67 & 0.09 & 1.80 & 1.69 &  1.97 & 0.09 & 2.06 & 1.90 &  2.20 & 0.09 & 2.60 & 2.54 &  2.70 & 0.05   \\
				\hline
				{200} & 1.25 & 1.20 &  1.28 & 0.03 & 1.57 & 1.50 &  1.65 & 0.06 & 1.78 & 1.69 &  1.84 & 0.05 & 2.03 & 1.92 &  2.10 & 0.05 & 2.58 & 2.53 &  2.65 & 0.04  \\
				\hline
				{300} & 1.25 & 1.19 &  1.30 & 0.03 & 1.59 & 1.48 &  1.67 & 0.07 & 1.82 & 1.73 &  1.93 & 0.07 & 2.08 & 2.02 &  2.15 & 0.05 & 2.63 & 2.58 &  2.68 & 0.03 \\
				\hline
				{400} & 1.25 & 1.19 &  1.31 & 0.03 & 1.59 & 1.53 &  1.64 & 0.04 & 1.80 & 1.74 &  1.86 & 0.04 & 2.10 & 2.04 &  2.15 & 0.03 & 2.63 & 2.55 &  2.66 & 0.03   \\
				\hline
				{500} & 1.25 & 1.21 &  1.27 & 0.03 & 1.59 & 1.53 &  1.62 & 0.03 & 1.82 & 1.75 &  1.89 & 0.05 & 2.08 & 2.02 &  2.16 & 0.05 & 2.62 & 2.59 &  2.68 & 0.03   \\
				\hline
				{600} & 1.25 & 1.21 &  1.29 & 0.02 & 1.59 & 1.54 &  1.64 & 0.04 & 1.80 & 1.73 &  1.88 & 0.05 & 2.07 & 2.02 &  2.11 & 0.02 & 2.63 & 2.61 &  2.65 & 0.01   \\
				\hline
				{700} & 1.24 & 1.21 &  1.27 & 0.02 & 1.57 & 1.55 &  1.59 & 0.01 & 1.79 & 1.71 &  1.84 & 0.04 & 2.08 & 2.04 &  2.11 & 0.02 & 2.63 & 2.60 &  2.66 & 0.02   \\
				\hline
				{800} & 1.24 & 1.23 &  1.26 & 0.01 & 1.59 & 1.55 &  1.62 & 0.02 & 1.80 & 1.73 &  1.88 & 0.05 & 2.09 & 2.05 &  2.14 & 0.03 & 2.62 & 2.59 &  2.67 & 0.03   \\
				\hline
				{900} & 1.25 & 1.22 &  1.28 & 0.02 & 1.59 & 1.54 &  1.66 & 0.04 & 1.80 & 1.75 &  1.86 & 0.04 & 2.08 & 2.02 &  2.17 & 0.04 & 2.63 & 2.60 &  2.66 & 0.02   \\
				\hline
				{1000} & 1.24 & 1.23 & 1.26 & 0.01 & 1.59 & 1.56 &  1.63 & 0.03 & 1.80 & 1.77 &  1.85 & 0.03 & 2.08 & 2.05 &  2.12 & 0.02 & 2.63 & 2.61 &  2.64 & 0.01   \\
				\hline
		\end{tabular}}
	\end{center}
	\caption{Density of the different instances of our graph benchmark.}\label{ta:density}
\end{table}

%\begin{figure}[tb]
%	\centering
%	\subfigure[]{
%		\includegraphics[width=0.45\textwidth]{ug_8_100_08_heur_visibility}
%		\label{fi:ug_8_100_08_heur_visibility}
%	}
%	\hfil
%	\subfigure[]{
%		\includegraphics[width=0.45\textwidth]{ug_8_100_08_opt_visibility}
%		\label{fi:ug_8_100_08_opt_visibility}
%	}
%	\caption{Two visibility representations of the same plane graph with $100$ vertices and $\rm p_{iv}=0.8$ computed by (a) \textsc{DrawHeurST} and (b) \textsc{DrawOptST}. The drawing computed by \textsc{DrawHeurST} has $14$ transitive edges, while that computed by \textsc{DrawOptST} has $7$ transitive edges. }\label{fi:ug_8_100_08_visibility}
%\end{figure}

%\begin{figure}[tb]
%	\centering
%	\subfigure[]{
%		\includegraphics[width=0.8\textwidth]{ug_6_100_05_heur_visibility}
%		\label{fi:ug_6_100_05_heur_visibility}
%	}
%	\hfil
%	\subfigure[]{
%		\includegraphics[width=0.8\textwidth]{ug_6_100_05_opt_visibility}
%		\label{fi:ug_6_100_05_opt_visibility}
%	}
%	\caption{Two visibility representations of the same plane graph with $100$ vertices and $\rm p_{iv}=0.5$ computed by (a) \textsc{DrawHeurST} and (b) \textsc{DrawOptST}. The drawing computed by \textsc{DrawHeurST} has $52$ transitive edges, while that computed by \textsc{DrawOptST} has $37$ transitive edges. }\label{fi:ug_6_100_05_visibility}
%\end{figure}

\end{document}